\documentclass[11pt]{article}
\usepackage{fullpage,amsfonts,amsmath,amsthm,amssymb,mathtools,mathrsfs,graphicx,upgreek,tablefootnote}
\usepackage[margin=1in,marginparwidth=0.75in]{geometry}
\usepackage[round]{natbib}

\usepackage[monochrome]{xcolor}

\usepackage{algorithm}
\usepackage[noend]{algpseudocode}

\usepackage{thm-restate}
\usepackage{epsfig}
\usepackage{bm}
\usepackage{comment}

\usepackage{enumitem}

\numberwithin{equation}{section}

\newtheorem{theorem}{Theorem}[section]
\newtheorem{proposition}[theorem]{Proposition}

\newtheorem{lemma}[theorem]{Lemma}

\theoremstyle{definition}
\newtheorem{definition}{Definition}

\newtheorem{remark}[theorem]{Remark}

\usepackage{bbm}

\usepackage{algpseudocode}
\usepackage{datetime}

\usepackage[english]{babel}
\usepackage[autostyle, english = american]{csquotes}
\MakeOuterQuote{"}

\newcommand{\scr}{\mathcal}
\newcommand{\mb}{\mathbb}

\newcommand{\E}{\mathbb E}
\newcommand{\eps}{\varepsilon}

\newcommand{\Var}{\mathrm{Var}}

\newcommand{\Cov}{\mathrm{Cov}}

\newcommand{\Ber}{\mathrm{Ber}}

\newcommand{\bI}{\mathbbm{1}}

\newcommand{\T}{\scr{T}}
\newcommand{\q}{q}
\usepackage[breakable, skins]{tcolorbox}

\newcommand{\revcolor}[1]{\textcolor{blue}{#1}}
\usepackage{color-edits}
\addauthor{Will}{blue}

\usepackage{hyperref}
\usepackage[nameinlink]{cleveref}
\crefname{subsection}{subsection}{subsections}

\begin{document}
\title{Online Matching and Contention Resolution for Edge Arrivals with Vanishing Probabilities}


\author{Will Ma
\thanks{Graduate School of Business and Data Science Institute, Columbia University, \texttt{wm2428@gsb.columbia.edu}}
\and
Calum MacRury
\thanks{Georgia Tech, School of Industrial and Systems Engineering (ISyE), 
\texttt{calum.macrury@isye.gatech.edu}}
\and
Pranav Nuti
\thanks{University of Chicago, Booth School of Business, \texttt{pranavn@stanford.edu}}
}
\date{}

\maketitle

\begin{abstract}
We study the performance of sequential contention resolution and matching algorithms on random graphs with vanishing edge probabilities.
When the edges of the graph are processed in an adversarially-chosen order, we derive a new OCRS that is $0.382$-selectable, attaining the "independence benchmark" from the literature under the vanishing edge probabilities assumption.
Complementary to this positive result, we show that no OCRS can be more than $0.390$-selectable, significantly improving upon the upper bound of $0.428$ from the literature.
We also derive negative results that are specialized to bipartite graphs or subfamilies of OCRSs. Meanwhile, when the edges of the graph are processed in a uniformly random order, we show that the simple greedy contention resolution scheme which accepts all active and feasible edges is $1/2$-selectable.
This result is tight due to a known upper bound.
{\color{blue}We then show that when the algorithm can choose the processing order, a slight tweak to the random order---give each vertex a random priority and process edges in lexicographic order---results in a strictly better contention resolution scheme that is $1-\ln(2-1/e)\approx0.510$-selectable. Moreover, we show that this bound is tight over any sequential contention resolution scheme, even one which may adaptively choose the order in which it processes edges. This provides a separation from the $0.544$ upper bound for offline contention resolution implied by the classic result of Karp and Sipser.}
Our positive results also apply to online matching on $1$-uniform random graphs with vanishing (non-identical) edge probabilities, extending and unifying some results from the random graphs literature.
\end{abstract}

\newpage

\section{Introduction}
{\color{blue}

A common problem faced by two-sided matching platforms such as ride-sharing apps and online labor marketplaces is to decide how to assign supply to demand arriving in an online fashion, while accounting for the value of the assignment and uncertainty about future demand \citep{ozkan2020dynamic,hu2022dynamic,feng2024two, pollner2022improved}. Since assigning a worker to one job prevents that worker from taking another job at the same time, platforms must compare the value of a match available now with the value of other matches that it might make later.

Such decisions can be modeled as \textit{Bayesian online matching problems}, where we have some distributional knowledge of the value of future matches, but we learn their actual realized values in an online fashion. A common strategy for these problems is to first solve a \textit{fluid (ex-ante) relaxation}, which makes fractional decisions maximizing the expected value subject to the feasibility constraints holding only in expectation (this is typically easy, since the relaxation is using a linear or convex program), and then find a way to convert its solution into an online policy \citep{alaei_2012,alaei2014bayesian,jiang2022tight,balseiro2024survey}.

Generally speaking, it is impossible for an online policy to match the value of the fluid relaxation due to the stochasticity of the realized arrivals. The goal is therefore to obtain some constant fraction of this value while satisfying the feasibility constraints in every realization. \textit{Contention resolution schemes} are a technical tool for achieving this goal \citep{feldman2021online}.

Motivated by these concerns, we study variants of contention resolution schemes for matchings on graphs. Formally, we have a graph $G=(V,E)$, and each edge $e\in E$ has a random state $X_e \in \{0,1\}$ that is unknown a priori, but known to be \textit{active} (i.e., $X_e =1$) independently with probability $x_e$. This $x_e$ represents the probability with which a fluid relaxation of an allocation problem would accept the edge.

The edges arrive one-by-one, at which point their random state is revealed.
If the edge is not active, then it is discarded;
if the edge is active, then an immediate decision must be made about whether to select the edge.
The edge is only feasible to select if it does not share a vertex with any previously-selected edge. That is, the set of edges selected must form a \textit{matching} in the graph.

The probabilities $(x_e)_{e\in E}$ form a fractional matching in the graph.
That is,
\begin{align} \label{eqn:introFracMatch}
\sum_{e\in\partial(v)} x_e &\le1   &\forall v\in V,
\end{align}
which says that the expected number of active edges incident to any vertex $v$ is at most 1. 
The selection problem can be interpreted as "rounding" a fractional matching into an integer one, where the rounding occurs online, and only active edges can be rounded up.
The goal is to provide a uniform guarantee that every edge $e\in E$ is selected with probability at least $cx_e$, for a constant $c$ as large as possible. An algorithm that provides such a guarantee is called a $c$-selectable contention resolution scheme.

To make the connection between contention resolution schemes and applications in Bayesian online matching problems precise, consider a model with a fixed set of jobs and workers, where each worker can complete at most one job. Vertices represent jobs and workers, and an edge $e=(u,v)$ represents assigning job $u$ to worker $v$.
We have distributional knowledge of $W_e$, the non-negative value obtained by making this assignment, and assume that the $W_e$ are independent of each other.
We learn the actual realized value of each pair in an online fashion and must immediately decide whether that pair should belong to the overall matching, while attempting to maximize the total value of the matching constructed at the end of the process.
This is the prophet inequality matching problem \citep{feldman2021online, GravinW23, Ezra_2020}.

For simplicity, we present the ex-ante relaxation for the setting when each $W_e$ is a weighed Bernoulli random variable;
that is, $W_{e} \in \{0, w_e\}$, where $\Pr[W_e = w_e] = p_e$ for some fixed value $w_e \ge 0$ and probability $p_e \in [0,1]$.
An ex-ante relaxation introduces a variable $y_e\in[0,1]$ for each edge $e$,
which should be thought of as the ideal probability of selecting edge $e$, given $W_e = w_e$. By also
imposing expected matching constraints on the vertices, this leads to the following linear program:
\begin{equation} \label{eqn:ex-ante}
\begin{aligned} 
\max\quad & \sum_{e\in E} w_e p_e y_e \\
\text{s.t.}\quad & \sum_{e\in\partial(v)} p_e y_e \le1 &&\forall v\in V, \\
& 0\le y_e\le 1 &&\forall e\in E.
\end{aligned}
\end{equation}
The relaxation upper bounds the expected value obtainable by any matching algorithm, including the ``prophet benchmark'' which is revealed the instantiations of $(W_e)_{e \in E}$ upfront. For an optimal solution to this program, say $(y_e)_{e \in E}$, notice that 
if $x_e = y_e x_e$, then the constraints on $(x_e)_{e \in E}$  are exactly~\eqref{eqn:introFracMatch}. 

We now explain how to use a contention resolution scheme for $(x_e)_{e \in E}$ to design a policy for the prophet inequality matching problem. When edge $e$ arrives, drawn a Bernoulli $Y_e$ of mean $y_e$ independently,
and declare edge $e$ \textit{active} if $Y_e =1$ and $W_e = w_e$. Notice that each edge is active independently with overall probability $x_e$,
and so we can apply a contention resolution scheme to decide whether to select $e$. As every edge $e$ is selected with probability at least $c$ conditional on being active,
we get
\[
\E\left[\sum_{e\in E} w_e \cdot \bm{1}_{e \in \scr{M}}  \right]\ge c\sum_{e\in E}w_e \cdot  x_e = c\sum_{e\in E}w_e \cdot  p_e \cdot y_e,
\]
where $\scr{M}$ is the matching selected by the scheme.
Thus, we have found an online algorithm that obtains a $c$ fraction of the ex-ante relaxation's value. While we have presented
the reduction for the special case of weighted Bernoulli's, an almost identical argument extends to the case
of arbitrary random variables by replacing \eqref{eqn:ex-ante} with a convex program (see \citep{feldman2021online, Lee2018} for details).

Similar reductions exist for sequential posted-price mechanisms \citep{ChawlaHMS10,yan2011mechanism, adamczyk2018random, pollner2022improved, derakhshan2026} and stochastic probing \citep{GuptaN13,GuptaNS16,GuptaNS17, borodin2023}.
In stochastic probing, an edge's state is revealed only after it is probed, and a probed active edge must be selected.
For matching applications, probing can represent offering a job to a worker, with additional \textit{patience constraints} limiting the number of offers a work is willing to consider, as originally introduced by \cite{chen,BansalGLMNR12}. These offers can also depend on some monetary payment \cite{pollner2022improved, derakhshan2026}, in which case the acceptance probability of the worker taking the job may depend on the payment offered.  
We defer to relevant works for further details on these reductions.

}

In {\color{blue}applications of} contention resolution, the order in which edges are processed is important.  This leads to the following problem variants.
\begin{itemize}
\item \textbf{Adversarial Order}:
the edges are processed in a worst-case\footnote{
We assume that this order is chosen by an adversary who knows the description of the algorithm. They cannot however change the order after the algorithm begins its execution.} order.
An online selection algorithm in this setting is called an Online Contention Resolution Scheme (OCRS).
\item \textbf{Random Order}:
the edges are processed in a uniformly random order.
An online selection algorithm in this setting is called a Random-order Contention Resolution Scheme (RCRS).
\item \textbf{Free Order}:
the edges are processed in an order chosen
by the algorithm.
An online selection algorithm in this setting is called a Free-order Contention Resolution Scheme (FO-CRS).
\end{itemize}
We note that contention resolution was originally studied in the offline setting \citep{chekuri2014submodular}, where all random states are revealed before making any selections.

In this paper we consider OCRS, RCRS, and FO-CRS in the \textit{vanishing regime}, which was previously resolved by \citet{nuti2023towards} in the offline setting.
The idea of this regime is to assume every edge is active with a probability less than some small $\eps$, and provide limiting guarantees as $\eps\to0$.
This is an important regime for the following reasons. 
\begin{itemize}
\item
\textbf{Precedent}: $x_e$ being small has also been a special case of interest in other models of online matching:
\begin{itemize}
    \item Literature on the online matching with stochastic rewards model \citep{MehtaP12,mehta2014online,huang2020online,goyal2023online,HuangJSSWZ23} generally focuses on this regime. 
    \item Motivated by online edge-coloring, rounding fractional matchings online with small $x_e$ has also been studied extensively. {\color{blue}In this problem, there is no active/inactive status associated with edges. Instead, the fractions $x_e$ are revealed in an online fashion, and when $x_e$ is revealed, we must immediately decide whether $e$ should be included in the matching. The goal is to select a matching which contains edge $e$ with probability at least $cx_e$.} \citet{cohen2019} showed that for bipartite one-sided vertex arrivals {\color{blue}(when a vertex arrives, the values of all its incident edges are revealed)}, a $1-o(1)$ guarantee is attainable, where $o(1)$ is a function that tends to $0$ as $\max_{e \in E} x_e \rightarrow 0$. For edge arrivals, \citet{kulkarni2022} proved a rounding guarantee of $1-1/e -o(1)$, and this was recently improved to $1-o(1)$ by \citet{blikstad2024}.
    \item The random graphs literature analyzes what happens when $|V|\to\infty$, with $x_e=p\to0$ for all $e$ if the graph is sparse, which is a special case of the vanishing regime with equal probabilities.  We relate our results to this literature at the end of \Cref{sec:introOutline}.
\end{itemize}

\item \textbf{Tightness}: The vanishing or "Poisson" regime is often where worst-case examples occur in online Bayesian selection, especially in problems like contention resolution where one is competing against {\color{blue}an ex-ante benchmark}.  The simplest example arises in RCRS for rank-1 matroids \citep{Lee2018}, where the guarantee $c$ cannot be larger than $1-1/e$ due to the vanishing regime. 
More specifically for matchings in graphs, three settings where optimal contention resolution schemes are known \revcolor{to} have tight examples that lie in the vanishing regime: offline monotone bipartite contention resolution \citep{BruggmannZ22}, OCRS for two-sided vertex arrivals \citep{Ezra_2020}, and, RCRS for two-sided bipartite vertex arrivals \citep{macrury2023randomorder}.

Our results suggest that a plausible strategy for optimal contention resolution is to start with an algorithm for the vanishing regime, and then add adjustments as necessary to {\color{blue} to obtain the optimal algorithm for the non-vanishing regime}, even though we acknowledge the second step may not be easy.
{\color{blue}\citet[Lemma~7, Appendix~A.1]{BruggmannZ22} already show how to reduce to the vanishing regime if the reduced graph can have parallel edges by splitting each edge $e$ into $k$ parallel copies with values $x_e/k$}; however, our results are for graphs \textit{without parallel edges}.
\end{itemize}
We are therefore interested in the following question:

\begin{displayquote}
   \emph{What is the optimal OCRS, RCRS, and FO-CRS in the vanishing regime?}
\end{displayquote}

\subsection{Outline of Results} \label{sec:introOutline}

We say that an OCRS (resp.\ RCRS, FO-CRS) is \textit{$c$-selectable} if it selects every edge $e\in E$ with probability at least $cx_e$, under an adversarial (resp.\ random, free) arrival order, for all graphs $G=(V,E)$ and vectors $(x_e)_{e\in E}$ satisfying~\eqref{eqn:introFracMatch}.
We say that it is \textit{$c$-selectable in the vanishing regime} if it is $(c-f(\max_{e\in E}x_e))$-selectable for all $G$ and $(x_e)_{e\in E}$ satisfying~\eqref{eqn:introFracMatch}, where $f$ is a function satisfying $\lim_{\eps\to0}f(\eps)=0$.
We call constant $c$ the \textit{selection guarantee} or just guarantee, which lies in [0,1].
We use \textit{selectability} to refer to largest $c$ for which an OCRS, RCRS, or FO-CRS can be $c$-selectable, possibly in the vanishing regime.
By \textit{optimal} contention resolution scheme, we refer to one that achieves the selectability for a particular setting (it does not have to be instance-optimal).

\begin{theorem} \label{thm:OCRS_positive}
There is a polytime OCRS that is $\frac{3 -\sqrt{5}}{2} \approx 0.382$-selectable
in the vanishing regime.
\end{theorem}

We let $\alpha$ denote the constant $\frac{3 -\sqrt{5}}{2}$, which is the smaller real number satisfying $\alpha=(1-\alpha)^2$.
This is the guarantee that would be achieved {\color{blue} if we already knew for a fact that the events that the endpoints of an arriving edge are unmatched were independent \citep[\S4.2]{Ezra_2020}.
We are able to show that this guarantee is in fact achieved in the vanishing regime.}
Formalizing the way in which small activeness probabilities $x_e$ imply approximate independence is the main technical challenge in this work, as we explain in \Cref{sec:technical}.
We also note that we are using a different OCRS for selecting edges than \citet{Ezra_2020}, because their algorithm has a sophisticated behavior whose independence properties are difficult to analyze. 
\begin{theorem} \label{thm:OCRS_negative}
\leavevmode

\begin{enumerate}

\item \revcolor{No OCRS can be better than $\gamma\approx0.393$-selectable, even in the vanishing regime over tree graphs.}
\label{enum:neg_tree}

\item \revcolor{No OCRS can be better than $\beta\approx0.390$-selectable, even in the vanishing regime.} 
\label{enum:neg_gen}

\item \revcolor{No OCRS that satisfies a concentration property (see \Cref{sec:OCRS_neg_trees} and \Cref{apx:concentration-property})} \label{enum:neg_conc} \revcolor{can be better than $\alpha=\frac{3 -\sqrt{5}}{2}$-selectable, even in the vanishing regime over tree graphs.}

\end{enumerate}
\end{theorem}
\revcolor{In parts \eqref{enum:neg_tree} and \eqref{enum:neg_gen} of \Cref{thm:OCRS_negative}, we establish upper bounds of $0.393$ and $0.390$ over vanishing tree graphs (which are bipartite graphs) and vanishing general graphs, respectively.}
Part \eqref{enum:neg_conc} of \Cref{thm:OCRS_negative} says that the `independence benchmark' of $\alpha$ is tight over OCRSs that satisfy a certain concentration property. This property holds for all OCRSs in the literature (we discuss this in greater detail in {\color{blue}\Cref{apx:concentration-property}}). {\color{blue}Moreover, if our objective was instead to maximize expected size of the algorithm's matching, then a
standard Doob martingale argument implies that the optimal policy is concentrated. Due to these points, we believe it is intuitive to posit that an optimal OCRS should satisfy this concentration property, even though we are not able to prove it.}
As we explain in \Cref{sec:technical}, these upper bounds use different arguments than existing ones in the literature, and also offer a significant improvement.
Indeed, the previous upper bounds were $4/9\approx0.444$ \citep{gravin2019prophet} followed by $3/7\approx0.428$ \citep{correa2022optimal} for bipartite graphs, and $0.4$ \citep{macrury2023random} for general graphs (this final result did not use vanishing edge values).

Turning to RCRS and FO-CRS, we will assume that $G$ is \textit{1-regular}, which makes it easier to describe optimal algorithms.
This is the assumption that $\sum_{e\in\partial(v)}x_e=1$ for all $v$, i.e.~\eqref{eqn:introFracMatch} is tight, which has been shown to be the worst case via reductions (see \Cref{sec:rcrs_pos}).
In the random order setting, {\color{blue} we show that the greedy \footnote{This is not to be confused with the notion of a greedy OCRS from \citet{feldman2021online}.}CRS, which accepts every arriving edge that is active and feasible is $1/2$-selectable:}

\begin{theorem} \label{thm:RCRS_positive}
The greedy CRS is a $\frac{1}{2}$-selectable RCRS in the vanishing regime.
\end{theorem}

\Cref{thm:RCRS_positive} is tight because no RCRS can be better than 1/2-selectable, even in the vanishing regime \citep{macrury2023random}.
It is worth noting that the 0.474-selectable RCRS of \citet{macrury2023random}, which is state-of-the-art for non-vanishing edges, behaves identically to the greedy algorithm in the vanishing regime, and hence is 1/2-selectable.
This contrasts with the OCRS variant, where the state-of-the-art algorithm originating from \citet{Ezra_2020} cannot attain the selection guarantee of $\alpha$ from the vanishing regime (its guarantee is upper-bounded by 0.361, as shown by \citet{macrury2023random}).

\begin{theorem}\label{thm:free_order_positive}
There exist random orderings of edges under which the greedy CRS is a $1 - \ln \left(2 - \frac{1}{e}\right) \approx 0.510$-selectable FO-CRS in the vanishing regime.
\end{theorem}

\Cref{thm:free_order_positive} establishes a strict separation between free and random order, for edge arrivals in graphs in the vanishing regime.
To our knowledge, results that show how to improve upon random order are generally rare. 
In stochastic probing and price of information problems on graphs in which the algorithm can choose the order, many results \citep{Gamlath2019,Fu2021} still defer to analyzing the random order setting at some point.

\begin{theorem}
\label{thm:fo-upper-bound}
{\color{blue}No FO-CRS can be better than 
better than
$
 1-\ln\left(2-\frac{1}{e}\right)
$
selectable in the vanishing regime. }
\end{theorem}

{\color{blue}
We prove \Cref{thm:fo-upper-bound} by showing that on
the complete bipartite graph with $2 n$ vertices and uniform edge values, 
no FO-CRS can construct a matching whose expected size is asymptotically greater than $(1 - \ln \left(2 - \frac{1}{e}\right)) n$ as $n \rightarrow \infty$. A classic result of \citet{karp1981maximum} showed that the optimal matching of this same input has expected size asymptotically $\approx 0.544 n$, and \cite{nuti2023towards} showed that there is an offline CRS with a selection guarantee of $\approx 0.544$ in the vanishing regime. Thus, \Cref{thm:fo-upper-bound} gives a strict separation between free-order contention resolution and offline contention resolution in the vanishing regime.}

\textbf{Relation to oblivious matching and query-commit.} In the edge-weighted \textit{oblivious matching problem}, the edges have known weights, however whether they exist or not is determined by an unknown (arbitrary) distribution.  The algorithm probes edges in an order of its choosing, subject to the constraint that if an edge is found to exist, it must be added to the matching. By randomly perturbing the edge weights and executing greedily, \citet{tang2023} prove a $0.501$-approximation ratio against the matching of largest expected weight. The oblivious matching setting generalizes the edge-weighted \textit{query-commit matching problem}, the latter of which has been studied extensively, beginning with the work of \citet{chen,BansalGLMNR12} (see \citet{borodin2023} for a detailed overview of this problem and its variants). In this problem, the edges are realized independently according to known probabilities. \citet{Fu2021} design a $0.533$-selectable random-order \textit{vertex arrival} OCRS which they use with the linear program of  \citet{costello2012matching, Gamlath2019} to get an approximation ratio of $0.533$ for this problem. (This selection guarantee, and thus approximation ratio, was recently improved to $0.535$ by \citet{macrury2023randomorder}). Assuming the edge probabilities are vanishing, we can apply our free-order OCRS to get a $0.510$ approximation ratio for this problem, and thus beat $1/2$ in a different way.

We summarize our results and how they position in the literature in \Cref{tab:my_label}.

\begin{table}[H]
\centering
\begin{tabular}{c|c|c} \hline
Selectability Bounds & General Edge Values $x_e$ & Vanishing Edge Values $x_e\to0$ \\ 
\hline
OCRS & $\ge 0.344$ [3] & $\to\ \ge \textbf{0.382}$ [\S\ref{sec:ocrs_positive}] \\ & $\le 0.4$ [3] $\to\ \le\textbf{0.390}$  [\S\ref{sec:OCRS_neg}]  & $\le 0.428$ [2] $\to\ \le\textbf{0.390}$, \\ && $\le 0.428$ [2] $\to\ \le\textbf{0.393}$ (bipartite) [\S\ref{sec:OCRS_neg}]  \\
\hline
RCRS & $ \ge 0.474$ [3]& $\to\ \ge\textbf{0.5}$ [\S\ref{sec:rcrs_pos}] \\
& $\le 0.5$ [3]  &  $\le 0.5$ [3] \\
\hline
FO-CRS & $\ge 0.474\ $ [3]  & $\to\ \ge \, \approx \textbf{0.510}$ [\S\ref{sec:rcrs_pos}] \\
& $\le 0.544$ [1] $\to\ \le \, \approx \textbf{0.510} $ [\S \ref{sec:free_order_negative}] & $\le 0.544$ [1] $\to\ \le \, \approx \textbf{0.510} $ [\S \ref{sec:free_order_negative}] \\
\hline
\end{tabular}
\caption{Summary of known results, with new results \textbf{bolded}.  "$\ge$" refers to lower bounds on $c$ (algorithmic results), "$\le$" refers to upper bounds (impossibility results), and arrows indicate improvement from state of the art. [1] is \citet{karp1981maximum}, [2] is \citet{correa2022optimal}, and [3] is \citet{macrury2023random}.
}
\label{tab:my_label}
\end{table}

\textbf{Relation to matching in random graphs.}
\citet{dyer1993average} show, among other results, that the greedy algorithm leaves 50\% of the vertices matched in expectation when it is executed on an Erdős–Rényi random graph with $n$ vertices and $n/2$ edges (i.e.\ average degree $1$), as $n\to\infty$, if it processes the edges in a uniformly random order.
They show that a modified order, in which a vertex is first chosen uniformly at random and then an edge incident to it (if any) is chosen uniformly at random, leaves $1 - \ln \left(2 - \frac{1}{e}\right)\approx51\%$ of the vertices matched.  Their arguments are combinatorial.
\citet{mastin2013greedy} show that these answers (under either order) are unchanged if one considers a complete bipartite Erdős–Rényi graph instead, using the differential equation method \citep{de}.
Our \Cref{thm:RCRS_positive,thm:free_order_positive} can be viewed as a stronger version of these results (in the special case with average degree 1), because:
\begin{itemize}
\item We allow for an \textit{arbitrary} graph with vanishing probabilities and average degree 1, whereas both the results for complete and complete bipartite Erdős–Rényi graphs assume that edge probabilities are equal;
\item Despite the asymmetry in our probabilities, we provide \textit{uniform} guarantees to every edge for being selected, and show that the guarantee does not get worse.  Put another way, our guarantee holds for max-weight matching instead of just max-cardinality matching.
\end{itemize}
As explained in \Cref{sec:technical}, we also use a differential equation method based on the Galton-Watson trees from \citet{nuti2023towards}.
Interestingly, the random edge orderings we consider in \Cref{thm:free_order_positive} are different\footnote{We draw a random priority for each vertex. The property we require in our analysis for the ordering of the edges is that an edge $(u, w)$ is processed before $(u, v)$ exactly when $w$'s priority is less than $v$'s priority.} from the modified order described above.


\Willedit{
\textbf{Further related work.}
In the classical revenue management formulation of a resource allocation problem \citep{gallego1994optimal}, and in many models of matching problems (see, for instance \citet{aouad2022dynamic} or \citet{ozkan2020dynamic}) demand arrives according to a Poisson process. Since a Poisson process of arrivals is a limit of tiny Bernoulli arrivals, results in the vanishing regime are related to problems with Poisson arrivals \citep[see][]{ma2024online, aminian2026stationary}. However, a formal reduction requires results for graphs with parallel edges.
}

\section{Technical Overview}
\label{sec:technical}

Before explaining the techniques used in proving our results, we formalize what we mean for an input $(G,\bm{x})$ to have vanishing edges values. Our definition includes the Erdős–Rényi random graph with vanishing edge probabilities case, but allows us to work with a sequence of inputs $(G(n), \bm{x}(n))_{n \ge 1}$ dependent on $n$ in an arbitrary way. 
\begin{definition} \label{def:vanishing_edge_values}
Given an input $(G,\bm{x}) =(G(n), \bm{x}(n))$ indexed by a parameter $n \in \mb{N}$, we say it has \textit{vanishing edge values}, provided there exists $(\eps(n))_{n \ge 1}$, such that $x_{e}(n) \le \eps(n)$ for each edge $e$ of $G(n)$, where
$\eps(n) \rightarrow 0$ as $n \rightarrow \infty$. In this case, we say that
$(G(n), \bm{x}(n))$ has \textit{vanishing edge values with respect to} $\eps(n)$. When the context
is clear, we drop the dependence on $n$ in our notation.

In addition, we say a contention resolution scheme is $c$-selectable on an input $(G,\bm{x})$ with vanishing edge values,
provided it is $c_n$-selectable on $(G(n), \bm{x}(n))$, where $c_n \rightarrow c$ as $n \rightarrow \infty$. 

Finally, we say that a sequence of events $(\scr{E}(n))_{n \ge 1}$ occurs \textit{with high probability} (w.h.p.), provided $\Pr[\scr{E}(n)] \rightarrow 1$ as $n \rightarrow \infty$.
\end{definition}

In the rest of this section, we will summarize the techniques used in proving our results. We begin with our positive results for
RCRS and FO-CRS, as the algorithms are very simple to describe, and the analysis is based on the Galton-Watson approach developed by \citet{nuti2023towards} for an offline CRS with vanishing edge values. {\color{blue}Afterwards, we explain our approach for upper-bounding a FO-CRS, and contrast it with the previously known upper bound of $1/2$ for RCRSs.}
We then discuss our positive result for OCRS, where we first explain why the Galton-Watson approach does not seem
applicable in this setting. Finally, we discuss the constructions used in proving our negative results for OCRS,
{\color{blue}and explain why the arguments differ from the RCRS/FO-CRS case.}

\paragraph{RCRS and FO-CRS positive results.}
Let us assume that $(G,\bm{x})$, with $G=(V,E)$ is a $1$-regular input, i.e., $\sum_{e \in \partial(v)} x_e =1$
for each $v \in V$. (This is without loss, as we explain in \Cref{sec:rcrs_pos}). Both positive results then analyze the \textit{greedy CRS} on $G=(V,E)$, which accepts any active edge whenever possible, where the only difference between the settings is the order in which the edges are processed. For RCRS, we assume each
edge $e$ has an \textit{arrival time} drawn independently and uniformly at random (u.a.r.) from $[0,1]$. For FO-CRS, we first draw a \textit{seed} from $[0,1]$ independently and u.a.r for each vertex. We then define $e=(u,v)$ to arrive at time $\min\{t_u,t_v\} + \eps \max\{t_u,t_v\}$, where $t_u$ and $t_v$ are the seeds of $u$ and $v$, and $\eps$ is any number with $\eps < \min_{u \neq v} | t_u - t_v|$. This ensures that the ordering is primarily decided by  $\min\{t_u,t_v\},$ and in case there is a tie, the ordering is then decided by $\max\{t_u,t_v\}$. We refer to this order as a \textit{lexicographical} random ordering.

For general edge values, the state-of-the-art is the RCRS of \citet{macrury2023random}. This algorithm
is very similar to the greedy CRS, except that edges with large $x_e$ are necessarily \revcolor{\textit{dropped} with probability $1 - a(x_e)$} for a carefully engineered function $a$. Since $a(0)=1$, this algorithm is asymptotically identical to the greedy CRS for vanishing edge values. Due to this connection, a natural approach to prove \Cref{thm:RCRS_positive} would be to improve the analysis of \citet{macrury2023random} for vanishing edge values.
Unfortunately, their analysis is inherently lossy---when lower bounding the probability $e$ is matched, it only considers vertices at distance \textit{at most} $3$ from $u$ or $v$.

We follow instead the approach of \citet{nuti2023towards}
and \textit{reduce} the problem to analyzing a greedy matching algorithm on a certain randomly generated tree. To explain this approach, fix $e =(u,v) \in E$ which is active, and define $S_e$ to be the connected component of $e$ induced by the \textit{active} edges of $G$. 
There are three essential observations underlying our analysis:
\begin{enumerate}
    \item As $n \rightarrow \infty$, $S_e$ looks like a (finite) random tree $\T$, constructed by drawing independent Galton-Watson trees $\T_u$ (at $u$), $\T_v$ (at $v$) with Poisson parameter $1$, and setting $\T = \T_u \cup \T_v \cup \{e\}$.
    \item Whether the greedy CRS selects a particular active edge $e$ depends entirely \revcolor{upon $S_e$}, and the order in which the edges of $S_e$ are presented to the algorithm. \label{enum:property_active}
    \item The edges of $S_e$ are themselves presented to the greedy CRS in a random order.
\end{enumerate}
Using these observations, we equate the probability the greedy CRS selects $e$ to the probability the greedy matching algorithm on $\scr{T}$ selects $e$, when presented the edges of $\scr{T}$ in a random order
(see \Cref{thm:reductiontogw}). Assuming this reduction, let us now sketch how the analysis on $\T$ proceeds for u.a.r. edge arrivals. 

We define $q(t)$  to be the probability the vertex $u$ in edge $e$ is not matched when $e=(u,v)$ arrives, given it arrives at time $t$. Note that due to the symmetry between $u$ and $v$, this is also the conditional probability $v$ is not matched. Then, 
since $\T_u$ and $\T_v$ are drawn independently, and $e$ is matched if only if $u$ and $v$ are \textit{not} matched when $e$ arrives,
$e$ is matched w.p. $\int_{0}^{1} \q(t)^2 dt$ after integrating over the arrival time of
$e$. By considering the children of $u$ in $\T_u$, and using the recursive nature of the Galton-Watson tree process, we then prove that $\q(t)$ satisfies a
certain integral equation. Using the fact that $q(0)=1$, we show that the integral equation has the unique solution $q(t) = 1/(t+1)$, and so $\int_{0}^{1} \q(t)^2 dt =1/2$ as desired.

For a lexicographical random ordering, while the analysis is similar in broad strokes, it is more involved,
as we have to consider the function $q(t_1,t_2),$ which is the probability vertex $u$ of $e$ is unmatched, given the seed of $u$ is $t_1$
and the seed of $v$ is $t_2$.

\paragraph{{\color{blue}FO-CRS negative result.}}
To prove \Cref{thm:free_order_positive},
we consider the complete bipartite graph $G$ with
$2 n$ vertices and uniform edge values on its $n^2$ edges. 
We prove the stronger statement, that if a FO-CRS outputs matching $\scr M$ on $G$,
then $\frac{\mb{E}[|\scr M|]}{n} \le (1 + o(1)) \left(1 - \ln \left(2 - \frac{1}{e}\right) \right)$
as $n \rightarrow \infty$. This is same objective taken for \textit{random-order edge arrivals} by \citet{macrury2023random}. For each step $t=0, \ldots, n^2$, \citet{macrury2023random} track
the matching $\scr{M}_t$ constructed by an RCRS after the first $t$ random edges are processed in the u.a.r arrival order.
Then, they argue that on a typical sample path history $\scr{H}_t$, $\mb{E}[ |\scr{M}_{t+1}| - |\scr{M}_t| \mid \scr{H}_t]$ is upper bounded by roughly $\frac{1}{n} \left( 1 - \frac{|\scr{M}_t|}{n} \right)^2$, where the expectation is over the endpoints of the $(t+1)^{th}$ edge, say $e_{t+1}$, and whether it is active.
Viewing $\mb{E}[ |\scr{M}_{t+1}| - |\scr{M}_t| \mid \scr{H}_t]$ as a ``derivative'', this allows them to apply the one-sided differential equation method of \cite{bennett2023} to conclude that $|\scr{M}_{n^2}|/n$ is roughly at most $w(1) =1/2$, where $w$
is the unique solution to the differential equation $w'(\tau) = (1 - w(\tau))^2$ with $w(0) = 0$.

The challenge in the the free-order setting is that since the endpoints of $e_{t+1}$ are chosen by the FO-CRS, there is no such simple upper bound on the one step expected change. Instead, we track the ``residual subgraph'' $R_t \subseteq G$ after the first $t$ edges are processed by the FO-CRS, which contains the remaining edges of $G$ which are compatible with its current matching $\scr{M}_t$. Then, we define a ``surplus function'' $\scr{S}$, which assigns a non-negative value $\scr{S}(R_t) \ge 0$ to $R_t$. 
Intuitively, $\scr{S}(R_t)$ is a potential function which is a proxy for the expected matching size available to an FO-CRS which starts with the subgraph $R_t \subseteq G$.
We then consider the random variable
$
    Q_t := |\scr{M}_t| + \scr{S}(R_t),
$
and argue that for an appropriately chosen $\scr{S}$, this sequence of random variables
form a supermartingale, which at $t=0$, takes the value $(1 + o(1)) \left(1 - \ln \left(2 - \frac{1}{e}\right) \right)n$. These two facts allow use to conclude that $\frac{\mb{E}[|\scr M|]}{n} \le (1 + o(1)) \left(1 - \ln \left(2 - \frac{1}{e}\right) \right)$ as desired (see \Cref{lem:assuming_residual}).

The choice of $\scr{S}$ is reverse engineered through our target upper bound of  $1-\ln\left(2-\frac{1}{e}\right)$,
and the technical properties required to argue that
the $Q_t$ form a supermartingale.
Specifically, if each vertex of $G$ is \textit{marked} independently with probability $1/n$,
then for a subgraph $H\subseteq G$, we set
\begin{equation*} 
    \scr{S}(H):=-n\ln(\Pr[\text{the marked vertices are independent in $H$}]).
\end{equation*}
It is not hard to see that $\scr{S}(G) = -n\ln\left(2(1-1/n)^n-(1-1/n)^{2n}\right) =
    \left(
        1-\ln\left(2-\frac{1}{\mathrm e}\right)+o(1)
    \right)n.$
The remaining properties required of $\scr{S}$ are some conditional probability
inequalities involving when the marked process is independent in $H$. 
We state and motivate them in greater detail in \Cref{sec:free_order_negative}. 

\paragraph{OCRS positive result.}
The techniques used in the proof of \Cref{thm:OCRS_positive} differ substantially from the RCRS/FO-CRS setting. This is because
the decisions of our OCRS depend on the active \textit{and} inactive edges of $G$. As such, \revcolor{property \eqref{enum:property_active} above} does \textit{not}
hold, and so we cannot reduce the problem to analyzing a matching algorithm on $\T$. The OCRS of \citet{Ezra_2020}, which provides
the state-of-the-art lower bound on OCRS selectability \citep{macrury2023random}, also has this limitation.

Before we describe our OCRS, we first review the OCRS of \citet{Ezra_2020}, and explain why it is challenging to
prove it is $\alpha \approx 0.382$-selectable on inputs with vanishing edge values (even if we abandon the Galton-Watson approach).
This OCRS is given a graph $G=(V,E)$ with an arbitrary fractional matching $\bm{x} = (x_e)_{e \in E}$. It then selects arriving edge $e$ with probability $c/\Pr[\text{$u$ and $v$ unmatched before $e$}]$, when $e$ is active and its endpoints are unmatched. This ensures it selects each edge $e \in E$ with probability \textit{exactly} $cx_e$, and the goal is to make $c$ as large as possible.

In order to prove this is well-defined, one must show $c \le \Pr[\text{$u$ and $v$ unmatched before $e$}]$, while assuming the induction hypothesis that every previous edge $f$ before $e$ was selected
w.p. $c x_f$. Now, if $G$ is a tree, then the events ``$u$ is matched before $e$'' and ``$v$ is matched before $e$''
are independent, and so the induction hypothesis
yields $\Pr[\text{$u$ unmatched before $e$}] \cdot \Pr[\text{$v$ unmatched before $e$}] \ge (1-c)^2$, so one can set $c = \alpha$, as $\alpha = (1- \alpha)^2$.

Suppose now that $(G,\bm{x}) = (G(n), \bm{x}(n))$ has vanishing edges values. \citet{nuti2023towards}
proved that $S_e$ is cycle-free w.h.p (here $S_e$ is defined as before). Thus, restricted to the active edges of $G$, the neighborhood of $e$ looks \textit{tree-like}. Suppose that conditional on $S_e$
being cycle-free, the events ``$u$ is unmatched before $e$'' and ``$v$ is unmatched before $e$'' are independent. Since $S_e$ occurs w.h.p., this would be enough to prove $\Pr[\text{$u$ unmatched before $e$}] \cdot \Pr[\text{$v$ unmatched before $e$}] \ge (1 -o(1)) (1-c)^2$,
and thus prove \Cref{thm:OCRS_positive}. However, this conditional
independence does \textit{not} hold for arbitrary inputs\footnote{Consider a graph on $3$ vertices $u,v,w$ with uniform edge values, where $(u,v)$ arrives last. Then $\Pr[\text{$u$ and $v$ matched before $(u,v)$}]=0$, yet $\Pr[\text{$S_{u,v}$ is cycle-free}] > 0$, and $\Pr[\text{$u$ matched before $e$} | \text{$S_e$ cycle-free}]$, $\Pr[\text{$v$ matched before $e$} | \text{$S_e$ cycle-free}] >0$.}, and so we do not analyze this OCRS.

We instead prove \Cref{thm:OCRS_positive} in two self-contained parts: \Cref{thm:OCRS_locally_tree_positive} and \Cref{thm:vanishing_to_local}. 
In
\Cref{thm:OCRS_locally_tree_positive}, we allow $\bm{x}$
of $G=(V,E)$ to be an arbitrary fractional matching. We say that an edge
$e$ is $\ell$-\textit{locally-tree-like}, provided there are no cycles in the $\ell$-neighborhood
of $e$. We design an OCRS (\Cref{alg:tree_aom}) for $G$ 
in which each active edge which is $\ell$-locally-tree-like is matched with probability $\alpha_\ell$,  where $\alpha_\ell$ converges to $\alpha \approx 0.382$ as $\ell \rightarrow \infty$ (see \eqref{eqn:alpha_k}
for the explicit definition of $\alpha_\ell$). Note that if $G$ has girth $g \ge \revcolor{5}$, then any edge is $\lfloor (g-\revcolor{3})/2 \rfloor$ locally-tree-like,
and so \Cref{alg:tree_aom} is $\alpha_{\lfloor (g-\revcolor{3})/2 \rfloor}$-selectable on such inputs. The convergence $\alpha_\ell \to \alpha$ occurs rapidly; in particular, $\alpha_7 \ge \revcolor{0.3561}$, and so \Cref{alg:tree_aom} beats the state-of-the-art of $0.344$ for inputs
with girth $g \ge \revcolor{17}$. 

\Cref{alg:tree_aom} is related to the OCRS of \citet{Ezra_2020}, however it is \textit{not}
defined recursively, and so it is easier to analyze. Specifically, when an active edge $e =(u,v)$ arrives and $u,v$ are unmatched, $e$ is added to the matching with probability $\frac{\alpha}{(1 - \alpha x_{u}(e))(1 - \alpha x_{v}(e))}$,
where $x_{v}(e)$ (respectively, $x_{u}(e)$) is the fractional degree of $u$ (respectively, $v$) prior to the arrival of $e$. 
Note that if $G$ is a tree, then \Cref{alg:tree_aom} is identical to the OCRS of \citet{Ezra_2020},
and so it is exactly $\alpha$-selectable. 

The core of the analysis is a \textit{correlation decay} argument in terms of $\ell$. We argue that since $e$ is $\ell$-locally-tree-like, there is no \textit{short} path between $u$ in $v$ in $G \setminus e$,
and so the events ``$u$ is matched before $e$'' and ``$v$ is matched before $e$''
are \textit{asymptotically} independent as $\ell \rightarrow \infty$. 
The analysis here is similar to the online fractional rounding algorithm of \citet{kulkarni2022}, with the main difference being that our analysis works for a general matching $\bm{x}$, whereas theirs requires $\bm{x}$ to have vanishing edge values (and in fact be uniform).

The second step of the proof is \Cref{thm:vanishing_to_local}, which provides a \textit{black-box reduction} from the \textit{vanishing edge value regime} to the \textit{locally-tree-like regime}. That is, given a CRS $\psi_{\ell}$ which is $c_\ell$-selectable for $\ell$-locally-tree-like edges, we show how to use it to design a new CRS which is $c$-selectable for inputs with edge values which are vanishing with respect to $\eps$, where $c= \lim_{\ell \rightarrow \infty} c_\ell$. 
The reduction proceeds via a ``two-round exposure'' argument, where we split the randomness of each edge state $X_e$ into two random variables, and analyze each separately.  We first set $y_e := \log(1/\eps) x_e$ and $z_e := 1/\log(1/\eps)$ for each $e \in E$. Since $\eps = o(1)$, $y_e, z_e \in [0,1]$, and so we can sample $Y_e \sim \Ber(y_e)$ and
$Z_e \sim \Ber(z_e)$ for each $e \in E$, such that $X_e = Y_e Z_e$, and for which $((Y_e)_{e \in E}, (Z_e)_{e \in E})$ are independent (see \Cref{prop:two_round_coupling}).

In the first round, we expose the $(Y_e)_{e \in E}$ random variables, and consider the sub-graph $G'$ of $G$ 
with edge-set $E'=\{e \in E: Y_e =1\}$.
We argue that for a specific choice of $\ell \rightarrow \infty$, any fixed edge of $G'$ is $\ell$-locally-tree-like w.h.p. (see \Cref{lem:cycle_upper_bound}). At this point, we'd like to expose the $(Z_e)_{e \in E'}$ random variables
and execute $\psi_{\ell}$ on $(G', (z_e)_{e \in E'})$. However, $(z_e)_{e \in E'}$ may not be a fractional matching of $G'$. Thus, we first greedily process the edges of $G'$, and create a fractional matching $\bm{z}' = (z'_e)_{e \in E'}$, where $z'_e \in \{0, z_e\}$ for each $e \in E'$. In \Cref{lem:good_fractional_matching}, we argue that for any fixed edge $e \in E'$, w.h.p. $z'_e = 1/\log(1/\eps)$ (i.e., $z'_e$ will \textit{not} be rounded down from $z_e$). Let us say an edge $e \in E'$ is \textit{well-behaved}, provided $e$ is $\ell$-locally-tree-like,
and $z'_e = 1/\log(1/\eps)$. In the second round, we expose the $(Z_e)_{e \in E'}$ random variables,
and execute $\psi_{\ell}$ on $(G',(z'_e)_{e \in E'})$ to ensure that each well-behaved edge $e \in E'$ is selected with probability $c_\ell z_e$. Since $e \in E$ is an edge of $G'$ with probability $y_e$, and w.h.p. $e \in E'$ is well-behaved, this ensures that $e$ is matched with an overall probability of $(1 - o(1)) c_\ell z_e y_e = (1-o(1)) c x_e$ after averaging over $G'$.

In order to execute our reduction online, we need to be able to compute $(z'_e)_{e \in E'}$ online. While this can't in general be done for an arbitrary OCRS $\psi_{\ell}$, it can be done if $\psi_{\ell}$ is \textit{strongly online}. Roughly speaking, the OCRS should work even if the fractional values of the matching are adversarially revealed. \Cref{alg:tree_aom} of \Cref{thm:OCRS_locally_tree_positive} is strongly online, and so \Cref{thm:OCRS_locally_tree_positive} and \Cref{thm:vanishing_to_local} together imply \Cref{thm:OCRS_positive}. 

\paragraph{OCRS negative result.} We prove the various results of \Cref{thm:OCRS_negative} by considering explicit graphs on which any OCRS cannot perform well. {\color{blue}Unlike our upper bound for FO-CRSs, we cannot just upper bound the expected size of the matching output. Instead, we must carefully construct graphs for which no matter the decisions of the OCRS, there exists some edge which cannot be taken with too high a conditional probability.}

For part \eqref{enum:neg_gen} of the theorem, consider a single vertex $u$ and suppose it has edges $\{(u,v_i)\}_{i=1}^n$ around it, with $x_{u,v_i} \approx 1/n$ for all $i$. For each $i \in [n]$, add a set $W_i$ of $n-1$ children to $v_i$, with $x_{v_i,w} \approx 1/n$ for each $w \in W_i$. Finally, we add a complete graph of edges amongst the $v_i$ with tiny edge values. 

All the edges $(v_i, w)$, $w \in \cup W_i$ are presented to the OCRS first. \revcolor{We say that $A_{v_i} = 1$ if $v_i$ is unmatched after these initial edges are presented to the OCRS.} The tiny edges are presented to the OCRS \textit{after} all the edges of the first phase but \textit{before} the edges $(u, v_i)$. The edges $(u, v_i)$ are presented to the OCRS last.

We consider two cases, depending on the size of the minimum pairwise covariance amongst the $A_{v_i}$ (see Lemmas \ref{lem:small_variance} and \ref{lem:large_variance}). Suppose first that after the first phase there is a pair $A_{v_i}$ and $A_{v_j}$ which are not very positively correlated with each other. Then, examining the newly added tiny edge between $v_i$ and $v_j$ tells us that no OCRS can perform particularly well. Indeed, if the scheme was $c$-selectable,
\begin{align*}
    c \leq \Pr[(v_i,v_j) \in \scr{M} \mid X_{v_i,v_j} =1] \le \Pr[A_{v_i} =1, A_{v_j} =1] &= \mb{E}[A_{v_i}] \mb{E}[A_{v_j}] + \Cov[A_{v_i},A_{v_j}] \\
    & \approx (1 - c)^2 + \text{something small}.
\end{align*}
Otherwise, let us suppose all the pairs $A_{v_i}$ and $A_{v_j}$ are positively correlated with each other, and we consider the last edge $(u, v_n)$ that is presented to the algorithm. Note first that it must be the case that if $A_{v_n} = 0$, then a higher than average number of the other $A_{v_i}$ are also zero. Conversely, a fewer than average number of the $A_{v_i}$ are one, since all the random variables are positively correlated.

Suppose that the OCRS does not impose any strong positive correlation between $A_{v_n}$ and the event $\mathcal{E}$ that $u$ is unmatched prior to observing $(u, v_n)$. In this case, exactly like the argument we just made, if the scheme was $c$-selectable, then
\begin{equation} \label{eqn:second_case_cov}
c \leq \Pr[\mathcal{E}, A_{v_n} = 1\mid X_{u,v_n} =1]= \Pr[\mathcal{E}]\E[ A_{v_n}]\approx (1-c)^2 + \text{something small}.
\end{equation}

So the OCRS must try to correlate $A_{v_n}$ with $\mathcal{E}$. But there is a limit to how positively correlated $A_{v_n} = 0$ can be with the event $ \neg \mathcal{E}$, because $ \neg \mathcal{E}$ can only occur if one of the edges $(u, v_i)$ with $A_{v_i} = 1$ is actually active, and there are a fewer than average number of $v_i$ with this property in the case we are considering. 

Next, consider part \eqref{enum:neg_conc} of the theorem. First, we would like to work with a tree, so we eliminate all the tiny edges we previous had between the $v_i$. But we also want to obtain an impossibility bound of exactly $\alpha$. So we need a new mechanism to ensure that $A_{v_n}$ (or perhaps some variation on $A_{v_n}$) has essentially no positive correlation with $\mathcal{E}$ (or perhaps some variation on $\mathcal{E}$).

We thus consider a \textit{concentrated} OCRS. An OCRS is said to be \textit{concentrated} on an input $(G,\bm{x})$, provided the matching $\scr{M}$ on $(G,\bm{x})$ it returns satisfies $||\scr{M}| - \mb{E}[ |\scr{M}|]| = o(\mb{E}[ |\scr{M}|)$ w.h.p. If the OCRS is concentrated, it must be the case that $\sum A_{v_i}$, while random, still exhibits concentration.

\revcolor{It follows then that} we can find a subset $S$ of size $m = o(n)$ such that $\sum_{v_i \in S} A_{v_i}$ is also concentrated. We prove this by taking a uniformly at random subset of size $m$, and arguing that there's a non-zero chance it has small variance.
Now, suppose that all these edges $(u, v_i)$ with $v_i \in S$ are presented to the OCRS last, and we imagine treating all of the edges as if they were a single `mega-edge'. It is impossible for the OCRS to strongly correlate $u$ being matched prior to this mega-edge being presented to the algorithm with $\sum_{v_i \in S} A_{v_i}$, since the latter sum is approximately deterministic. We use this fact to argue that any concentrated OCRS is, in fact, at most $\alpha$-selectable.

Finally, to establish part \eqref{enum:neg_tree} of the theorem, we start by arguing that there is a rearrangement of the $v_i$ so that $A_{v_n}$ is positively correlated with $\sum A_{v_i}$. Then, it must be the case that given $A_{v_n} = 0$, a higher than average number of the other $A_{v_i}$ are also zero. The rest of the argument is then very similar to the argument we used to prove \eqref{eqn:second_case_cov} in part \eqref{enum:neg_gen} of the theorem.

\section{OCRS Positive Result} \label{sec:ocrs_positive}

As discussed in \Cref{sec:technical}, there are two separate self-contained parts to the proof of \Cref{thm:OCRS_positive}. 
Let $\alpha \in (0,1)$ be the unique solution to the equation
\begin{equation} \label{eqn:alpha}
    \alpha = (1 - \alpha)^2,
\end{equation}
where $\alpha \approx 0.382$. For each $\ell \ge 1$,
we define
\begin{equation} \label{eqn:alpha_k}
\alpha_\ell :=\left(1 - \left(\revcolor{1-\alpha} \right)^{\revcolor{\ell}} \right)^2 \alpha,
\end{equation}
which we observe satisfies $\alpha_\ell \rightarrow \alpha$ as $\ell \rightarrow \infty$. Recall that for an edge $e=(u,v)$, we define the $\ell$-neighborhood of $e$ to be those vertices of $G$ at graph distance $\le \ell$ from $u$ or $v$, which we denote by $N^{\ell}(e)$. \revcolor{We also define $D^0(e):=N^0(e)$ and $D^i(e):=N^i(e)\setminus N^{i-1}(e)$ for each $i \ge 1$.} We prove
that if $N^{\ell}(e)$ is cycle-free, then $e$ is matched by \Cref{alg:tree_aom} with probability at least $\alpha_\ell x_e$. This motivates the following definition:

\begin{definition} \label{def:locally_tree_like}
Suppose that $\ell \in \mb{N}$, $c_\ell \in [0,1]$, and $\psi$ is a CRS for the matching polytope, whose matching on $(G,\bm{x})$ we denote by $\psi(G,\bm{x})$. We say that $\psi$ is \textit{$c_\ell$-selectable for $\ell$-locally-tree-like edges}, provided the following holds for any input $(G,\bm{x})$ where $G=(V,E)$:
\begin{itemize}
    \item For each $e \in E$, if $N^{\ell}(e)$ is cycle-free, then $\Pr[e \in \psi(G,\bm{x}) \mid X_e =1] \ge c_\ell$.
\end{itemize}
\end{definition}

\begin{theorem} \label{thm:OCRS_locally_tree_positive}
For each $\ell \ge 1$, \Cref{alg:tree_aom} is $\alpha_\ell$-selectable for $\ell$-locally-tree-like edges, where $\alpha_\ell$ is defined in \eqref{eqn:alpha_k}.
\end{theorem}

In order to state our reduction for the second part of the proof of \Cref{thm:OCRS_positive}, we require the OCRS to satisfy an additional property: When processing an edge
$e$, the decision to match $e$ depends on only $(x_e, X_e)$,
as well as the previous edges' fractional values and states (as well as any potential internal randomization used). Thus, the OCRS initially only knows $V$ of $G$,
and learns both $X_e$ \textit{and} $x_e$ of each edge $e$ in an adversarial order. This is in contrast to the usual definition of an OCRS, where it is assumed that the edges of $G$ and their edges values are known upfront. Note that we still work with an oblivious adversary who decides on the arrival order of $E$ prior to drawing $(X_e)_{e \in E}$.
We refer to an OCRS with this property as \textit{strongly online}.\footnote{{\color{blue}The OCRS of Ezra et al. is also strongly online. We do not know whether there is a separation between strongly online OCRSs and usual OCRSs.}} With this definition,
we can now  state our reduction.

\begin{theorem} \label{thm:vanishing_to_local}
Suppose $(\psi_\ell)_{\ell \ge 1}$ is a sequence of CRSs, where each $\psi_\ell$
is $c_\ell$-selectable on $\ell$-locally-tree-like edges, and $c := \lim_{\ell \rightarrow \infty} c_\ell$.
The following then hold:
\begin{enumerate}
    \item There exists an offline CRS $\psi$ which is $c$-selectable on graphs with vanishing edge values.
    \item If each $\psi_\ell$ is strongly online, then $\psi$ can be implemented as an OCRS.
\end{enumerate}

\end{theorem}
Since \Cref{alg:tree_aom} is a strongly online CRS, and $\alpha = \lim_\ell \alpha_\ell$, Theorems \ref{thm:OCRS_locally_tree_positive}
and \ref{thm:vanishing_to_local} together imply \Cref{thm:OCRS_positive}.
The rest of the section is devoted to proving each individual theorem.

\subsection{OCRS for Locally-Tree-Like Edges: Proving \Cref{thm:OCRS_locally_tree_positive}}
Given $G=(V,E)$ with fractional matching $\bm{x} = (x_e)_{e \in E}$, assume the edges $E$ of $G$ arrive in some adversarial order, where $e' < e$ indicates $e'$ arrives before $e$. Let $x_{v}(e):= \sum_{\substack{f < e: \\ f \in \partial(v)}} x_f$ for each $v \in V$ and $e \in \partial(v)$. Observe that $x_v(e)$ is the fractional degree of $v$, prior to the arrival of $e$.
\begin{algorithm}[H]
\caption{Tree-OCRS} 
\label{alg:tree_aom}
\begin{algorithmic}[1]
\Require $G=(V,E)$, $\bm{x}=(x_e)_{e \in E}$, and $\alpha$ which satisfies \eqref{eqn:alpha}.
\Ensure subset of active edges forming a matching $\scr{M}$
\State $\scr{M} \leftarrow \emptyset$
\For{arriving $e =(u,v) \in E$ with edge state $X_e$}
\State Draw $A_e \sim \Ber\left(\frac{\alpha}{(1 - \alpha x_{u}(e))(1 - \alpha x_{v}(e))}\right)$ independently. 
\If{$u$ and $v$ are unmatched, $X_e =1$ and $A_e =1$}
\State $\scr{M} \leftarrow \scr{M} \cup \{(u,v)\}$
\EndIf
\EndFor
\State \Return $\scr{M}$
\end{algorithmic}
\end{algorithm}
\begin{remark}
\Cref{alg:tree_aom} is well-defined, as $\max\{x_{v}(e), x_{u}(e)\} \le1$,
and $\alpha/(1 - \alpha)^2= 1$ by \eqref{eqn:alpha}. This is the only place
we use the specific value of $\alpha$. \Cref{alg:tree_aom} is well-defined
if it was executed with $c \leq \alpha$ instead of $\alpha$.
\end{remark}
In order to prove \Cref{thm:OCRS_locally_tree_positive}, we must show that
if $\scr{M}$ is returned by \Cref{alg:tree_aom}, then for each $e \in E$ with $N^{\ell}(e)$ cycle-free,
\begin{equation} \label{eqn:tree_ocrs_guarantee}
\Pr[e \in \scr{M} \mid X_e =1] \ge \alpha_\ell = \left(1 - \left(\revcolor{1-\alpha} \right)^{\revcolor{\ell}} \right)^2 \alpha.
\end{equation}
Fix $e=(u,v)$, and let $T$ be the graph induced by the vertices in
the $\ell$-neighborhood of $e$. Note that by assumption, $T$ is cycle-free
and connected, and so a tree. Denote by $\partial T$ the edges between 
$\revcolor{D^{\ell}(e)}$ and $\revcolor{D^{\ell+1}(e)}$, which we call the \textit{boundary} of $T$. Let us assume it is non-empty, since otherwise, the neighborhood of $e$ is a tree, and there is nothing to prove.

One difficulty with establishing guarantees on the selectability of an OCRS is the potentially arbitrary order in which edges might be presented to the algorithm. However, when studying \textit{one particular edge} $e$, we may assume that the edges appear in a much more specifc order.  We begin with a preliminary lemma (closely related to the `tree-matching game' discussed in \citet{kulkarni2022}) that demonstrates that this is the case:

\begin{restatable}[proven in \S\ref{pf:lem:tree_reduction_game_part_1}]{lemma}{treeReductionGamePartOne} \label{lem:tree_reduction_game_part_1}
When lower bounding $\Pr[e \in \scr{M} \mid X_e =1]$, we may assume that the edges of $T \cup \partial T$ are processed ``bottom up'': The edges between
        $\revcolor{D^{i}(e)}$ and  $\revcolor{D^{i+1}(e)}$ are processed before the edges between $\revcolor{D^{i-1}(e)}$
        and $\revcolor{D^{i}(e)}$ for $i=1, 2, \ldots , \ell$.
\end{restatable}

For the remainder of the proof, we assume that the order of the edges in $T$ and $\partial(T)$ is as described in \Cref{lem:tree_reduction_game_part_1}. This lemma is immensely useful, because it lets us come up with a recurrence letting us calculate the value of $\Pr[e \in \scr{M} \mid X_e =1]$. In order to establish this recurrence, first, consider vertices $s \in \revcolor{D^{\ell}(e)}$, and define the random variables $Q_s$, with $Q_s = 1$ if the OCRS (after making decision on the edges in $\partial(T)$) chooses to leave $s$ unmatched, and $0$ otherwise.

Suppose we fix $r \in T$, $r \notin \revcolor{D^\ell(e)}$ whose children (in the order they will be processed) are  $r_1, \ldots , r_k$ for $k \ge 0$ (if $k=0$, then $r$ has no children). We can explicitly write out the probability that a vertex $r$ is unmatched after $e_1=(r,r_1),\ldots , e_k=(r,r_k)$ are processed, conditional on the values of $Q_s$.
Denote this probability by $q_r$.
The key point is that since $T$ is a tree, and edges are processed in a bottom up fashion, once we condition on any particular values of $Q_s$, the children $r_1, \ldots ,r_k$ of $r$ are matched with their children independently. Thus,
\begin{align*}
    q_r &= \prod_{i=1}^{k} \Pr[e_i \notin \scr{M} \mid e_j \notin \scr{M}, \forall j < i] \\ &= \prod_{i=1}^{k} \left( 1 -\Pr[e_i \in \scr{M} \mid e_j \notin \scr{M},  \forall j < i] \right) \\
    &= \prod_{i=1}^{k} \left( 1 -\Pr[r_i \text{ is not matched with any of its children, } X_{e_i} = 1,  A_{e_i} = 1 \mid e_j \notin \scr{M}, \forall j < i] \right) \\
    &= \prod_{i=1}^{k} \left( 1 -\Pr[r_i \text{ is not matched with any of its children, } X_{e_i} = 1,  A_{e_i} = 1] \right) \\
    &= \prod_{i=1}^{k} \left( 1 -\Pr[r_i \text{ is not matched with any of its children}]\Pr[X_{e_i} = 1] \Pr[A_{e_i} = 1] \right) \\
    &= \prod_{i=1}^{k} \left( 1 - \frac{q_{r_i} x_{e_i} \alpha}{(1 - \alpha x_{r}(e_i))(1-\alpha x_{r_i}(e_i))} \right).
\end{align*}
So,
\begin{equation} \label{eqn:matched_tree_recurrence_game}
    q_{r} = \prod_{i=1}^{k} \left( 1 - \frac{q_{r_i} x_{e_i} \alpha}{(1 - \alpha x_{r}(e_i))(1-\alpha x_{r_i}(e_i))} \right),
\end{equation}
where $x_{r}(e_i)= \sum_{j=1}^{i-1}x_{r,r_j}$. If $r$ has  no children, then $q_r =1$. We implicitly understand in the above expression that when we write $q_s$ for $s \in \revcolor{D^\ell(e)}$, we actually mean $Q_s$. 

Let $x_r = \sum_{j=1}^k x_{r,r_j}$ for convenience, where $x_r = 0$ if $k=0$ (as $x_r$ never includes the parent of $r$). We wish to
show that if $r$ is ``high-enough'' up the tree, then $q_r \approx (1 - x_{r} \alpha)$.
To this end, define $\eps_{r} = 1 - \frac{q_{r}}{1 - \alpha x_r}$. Our goal is to convert the recurrence \eqref{eqn:matched_tree_recurrence_game} on the $q_r$ into a recurrence on the $\varepsilon_r$. This proceeds similarly to \citet{kulkarni2022}, though we include the derivation for completeness. We first use the equation $q_{r_i} = (1 - \eps_{r_i})(1 - \alpha x_{r_i})$ for $i=1, \ldots ,k$. By applying this to the right-hand of \eqref{eqn:matched_tree_recurrence_game},
\begin{align*}
    q_r = \prod_{i=1}^{k} \left( 1 - \frac{q_{r_i} \alpha x_{e_i}}{(1 - \alpha x_{r}(e_i))(1-\alpha x_{r_i}(e_i))} \right) &= \prod_{i=1}^{k} \left( 1 - \frac{(1 - \eps_{r_i})(1 - \alpha x_{r_i}) \alpha x_{e_i} }{(1 - \alpha x_{r}(e_i))(1-\alpha x_{r_i}(e_i))} \right) \\
    &= \prod_{i=1}^{k} \left( 1 - (1 - \eps_{r_i})\frac{\alpha x_{e_i}}{(1 - \alpha x_{r}(e_i))} \right) \\
    &= \prod_{i=1}^{k} \left( \frac{1 - \alpha x_{r}(e_i) - \alpha x_{e_i} + \eps_{r_i} \alpha x_{e_i}}{1 - \alpha x_{r}(e_i)} \right) \\
    &=\prod_{i=1}^{k} \left( \frac{1 - \alpha \sum_{j \le i} x_{r,r_j} + \eps_{r_i} \alpha x_{e_i}}{1 - \alpha \sum_{j < i} x_{r,r_j}} \right) \\
    &= \prod_{i=1}^{k} \frac{1 - \alpha \sum_{j \le i} x_{r,r_j}}{1 - \alpha \sum_{j < i} x_{r,r_j}}\left( 1  + \frac{\eps_{r_i} \alpha x_{e_i}}{1 - \alpha \sum_{j \le i} x_{r,r_j}} \right) \\
    & = (1 - \alpha x_r) \prod_{i=1}^{k}\left( 1  + \frac{\eps_{r_i} \alpha x_{e_i}}{1 - \alpha \sum_{j \le i} x_{r,r_j}} \right).
\end{align*}
Thus, applied to \eqref{eqn:matched_tree_recurrence_game}, since $\eps_{r} = 1 - \frac{q_{r}}{1 - \alpha x_r}$, we get that
\begin{equation} \label{eqn:error_recurrence}
    \eps_r = 1 - \prod_{i=1}^{k}\left( 1  + \frac{\eps_{r_i} \alpha x_{e_i}}{1 - \alpha \sum_{j \le i} x_{r,r_j}} \right).
\end{equation}
Note that for $r \in \revcolor{D^{\ell}(e)}$, since $q_r \in \{0,1\}$, $\eps_r = 1$ if
$q_r =0$, and $\eps_r = -\alpha x_r/(1-\alpha x_r) \ge - \revcolor{(1-\alpha)}$ if $q_r = 1$. Finally,
$\eps_r = 0$ if $r$ has no children.

Given $0 \le j \le \ell$, define the \textit{maximum error} at distance $\ell - j$
from $e$ by $\eps^{\max}_{j}:= \max_{r \in \revcolor{D^{\ell- j}(e)}}|\eps_r| $. We
argue that the maximum error \textit{decreases} the larger $j$ becomes (and so higher up the tree we go).
\begin{restatable}[proven in \S\ref{pf:lem:error_decrease}]{lemma}{errorDecrease} \label{lem:error_decrease}
For each $0 \le j \le \ell -1$, $\eps^{\max}_{j+1} \le  \eps^{\max}_{j} \left(\revcolor{1-\alpha}\right)$. Thus, $\eps^{\max}_{\ell} \le \left(\revcolor{1-\alpha}\right)^{\revcolor{\ell}}$,
since $\eps^{\max}_0 \le 1$.
\end{restatable}

\begin{proof}[Proof of \Cref{thm:OCRS_locally_tree_positive}]
It suffices to prove \eqref{eqn:tree_ocrs_guarantee}.
Fix any values for $(Q_s)_s$. First observe that 
$$\Pr[e \in \scr{M} \mid X_e =1\revcolor{,(Q_s)_s}] = \Pr[\text{$u$ and $v$ unmatched when $e$ arrives}\revcolor{\mid (Q_s)_s}] \frac{\alpha }{(1 - \alpha  x_{u}(e))(1 - \alpha  x_{v}(e)}.$$
\revcolor{Conditional on $(Q_s)_s$, the events that $u$ and $v$ are unmatched are independent, so the first probability equals $q_uq_v$.} But
$q_{u} = (1 - \alpha x_{u}(e))(1-\eps_u)$ and $q_{v} = (1 - \alpha x_{v}(e))(1-\eps_v)$. Thus,
since $\max\{\eps_{u}, \eps_{v}\} \le \left(\revcolor{1-\alpha}\right)^{\revcolor{\ell}}$ by \Cref{lem:error_decrease}, it follows that $\Pr[e \in \scr{M} \mid X_e =1\revcolor{,(Q_s)_s}] \ge (1 -  \left(\revcolor{1-\alpha}\right)^{\revcolor{\ell}})^{2} \alpha$ after
cancellation. \revcolor{Averaging over $(Q_s)_s$ proves} \eqref{eqn:tree_ocrs_guarantee}. The proof of \Cref{thm:OCRS_locally_tree_positive} is thus complete. 
\end{proof}

\subsection{Reduction to Locally-Tree-Like Edges: Proving \Cref{thm:vanishing_to_local}}
We shall make use of the following elementary coupling. 
\begin{restatable}[proven in \S\ref{pf:prop:two_round_coupling}]{proposition}{coupling}  \label{prop:two_round_coupling}
Suppose $x,y,z \in (0,1)$ satisfy $x \ge y z$. There exists a coupling
of Bernoulli random variables $X \sim \Ber(x)$, $Z \sim \Ber(z)$, and $Y \sim \Ber(y)$,
where $Y, Z$ are independent, and $X \ge Y \cdot Z$.
\end{restatable}

Suppose that $(G,\bm{x})=(G(n), \bm{x}(n))$ satisfies $x_e \le \eps$ for all $e \in E$, where $\eps = \eps(n)$ satisfies $\eps = o(1)$. Fix $\ell = \ell(n)$ where $\ell(n) = \log\log(1/\eps)$. Observe that $\ell(n) \rightarrow \infty$ as $n \rightarrow \infty$, since $\eps = o(1)$. We describe a CRS $\psi$ for $(G,\bm{x})$ which makes use of the CRS $\psi_\ell$ assumed in \Cref{thm:vanishing_to_local}. 
The idea is to first write $y_e = (1 - 1/\log^{1/4}(\eps^{-1})) x_e \log(1/\eps)$ and $z_e = 1/\log(1/\eps)$, where we observe that
$x_e \ge y_e z_e$ for each $e$.  There are then $3$ steps to $\psi$, which we describe below:
\begin{enumerate}
    \item \label{enum:subgraph_sample} Sample $G' \subseteq G$: Apply \Cref{prop:two_round_coupling} to each $e \in E$,
    to get $(Y_e)_{e \in E}$ and $(Z_e)_{e \in E}$ such that $X_e \ge Y_e Z_e$,
and for which $((Y_e)_{e \in E}, (Z_e)_{e \in E})$ are independent. Let $G'=(V,E')$ be the subgraph of $G$ where $E':=\{e \in E: Y_e =1\}$. 
    \item \label{enum:fractional_matching} Compute a fractional matching $(z'_e)_{e \in E}$ of $G$ supported on $E'$: Initialize $z'_e :=0$ for each $e \in E$. Then, process the edges $e=(u,v)$ of $E$ in an arbitrary order (such as the one presented by an oblivious adversary if we wish $\psi$ to be an OCRS), and update $z'_{e} \leftarrow 1/\log(1/\eps)$, if the (current) fractional matching constraints on $u$ and $v$ are not violated. (See \Cref{alg:fractional_matching} below for a formal description of this step).
    \item \label{enum:execute_tree_local} Run $\psi_{\revcolor{\ell}}$ on $(G', \bm{z}')$ where $\bm{z}' = (z'_e)_{e \in E'}$: For each $e \in E$, define $Z'_e = Z_e \cdot \bm{1}_{[z'_e > 0]}$. Execute $\psi_{\revcolor{\ell}}$ on $(G', \bm{z}')$ using $(Z'_e)_{e \in E'}$, and return the matching $\psi_{\revcolor{\ell}}(G', \bm{z}')$.
\end{enumerate}
Let us denote $\scr{M} = \psi_{\ell}(G', \bm{z}')$ to be the matching output by $\psi$. First observe that $\scr{M}$
only contains active edges of $G$, and so $\psi$ is indeed a valid CRS. To see this, observe that a necessary
condition for $e \in \scr{M}$ is that $Y_e Z'_e =1$. But since $X_e \ge Y_e Z_e$ by \Cref{prop:two_round_coupling},
and $Z'_e := Z_e \cdot \bm{1}_{[z'_e > 0]}$, it follows that $X_e =1$. 
\begin{algorithm}[H]
\caption{Fractional-Matching} 
\label{alg:fractional_matching}
\begin{algorithmic}[1]
\Require $G=(V,E)$ and $\eps \in [0,1]$.
\Ensure a fractional matching $\bm{z}'=(z'_e)_{e \in E}$ of $G$ supported on $E'$.
\State Set $z'_e = 0$ for all $e \in E$.
\For{arriving $e =(u,v) \in E$ in an arbitrary order $<$}
\If{$Y_{e} =1$, and $\max\{\sum_{\substack{f \in \partial_G(v):\\ f < e}} z'_f, \sum_{\substack{f \in \partial_{G}(u):\\ f < e}} z'_f\} \le 1 - \frac{1}{\log(1/\eps)}$}
\State $z'_e \leftarrow \frac{1}{\log(1/\eps)}$.
\EndIf
\EndFor
\State \Return $\bm{z}'=(z'_e)_{e \in E}$.
\end{algorithmic}
\end{algorithm}
\begin{remark}
\Cref{alg:fractional_matching} is stated and analyzed for an arbitrary edge ordering. When proving the OCRS part
of \Cref{thm:vanishing_to_local}, we will take this to be the one chosen by an oblivious adversary.
\end{remark}
To analyze $\psi$, we prove three lemmas. The first is a technical point
that ensures that once we condition on $G'$, the random variables $(Z'_e)_{e \in E}$
are independent. This ensures that we are indeed passing a valid CRS input to $\psi_{\ell}$
in step \ref{enum:execute_tree_local}. This follows from the fact that $(z'_e)_{e \in E}$ is a function of
$G'$, and $((Y_e)_{e \in E}, (Z_e)_{e \in E})$ are independent by construction.
\begin{restatable}[\revcolor{proven in \S\ref{pf:lem:conditional}}]{lemma}{conditionalIndependence}\label{lem:conditional}
Conditional on $G'$ (equivalently $(Y_e)_{e \in E}$), the random variables $(Z'_{e})_{e \in E}$ are independently distributed
with marginals described by $(z'_e)_{e \in E}$.
\end{restatable}

We next show that if $Y_e =1$, then w.h.p. $z'_e = 1/\log(1/\eps)$. In other words,
$z'_e$ will typically be rounded up to the desired value, assuming $e$ is an edge
of $G'$.
\begin{restatable}[proven in \S\ref{pf:lem:good_fractional_matching}]{lemma}{goodFractionalMatching}  \label{lem:good_fractional_matching}
For each $e \in E$, 
$\Pr[z'_e = 1/\log(1/\eps) \mid Y_e =1] \ge 1 - 2\exp\left(- \Omega(\log^{1/2}(1/\eps))\right)$.
\end{restatable}

Define $C_{\ell}(e)$ to be the event that there \textit{exists} a cycle in the
$\ell$-neighborhood of $e$ in $G'$. We next upper bound the probability that $C_{\ell}(e)$ occurs, conditional on $e$ being in $G'$ (i.e., $Y_e =1$).
\begin{restatable}[proven in \S\ref{pf:lem:cycle_upper_bound}]{lemma}{cycleUpperBound}  \label{lem:cycle_upper_bound}
For each $e \in E$, $\Pr[C_{\ell}(e) \mid Y_e =1] = O( \eps ( \log(1/\eps))^{2\ell+1})$.
\end{restatable}
We are now ready to prove \Cref{thm:vanishing_to_local}.
\begin{proof}[Proof of \Cref{thm:vanishing_to_local}]
In order to analyze $\psi$ on $G$, fix $e=(u,v) \in E$. Now, 
$e \in \scr{M}$ if and only if $e \in \psi_{\revcolor{\ell}}(G', \bm{z}')$. We shall first
show that
\begin{equation} \label{eqn:averaged_psi_performance}
\Pr[e \in \psi_{\revcolor{\ell}}(G', \bm{z}')] \ge (1 - o(1)) c x_e,
\end{equation}
Recall that $C_{\ell}(e)$ is the event that there \textit{exists} a cycle in the
$\ell$-neighborhood of $e$ in $G'$. Observe that $C_{\ell}(e)$ is a function of $(Y_f)_{f \in E}$ (i.e., it is $(Y_f)_{f \in E}$-measurable). Moreover, since $\psi_\ell$ is $c_\ell$-selectable on $\ell$-locally-tree-like edges by assumption,
we can apply \Cref{lem:conditional} to get that
\begin{equation*} 
    \Pr[e \in \psi_{\ell}(G', \bm{z}') \mid (Y_f)_{f \in E \setminus e}, Y_e =1] \cdot \bm{1}_{\neg C_{\ell}(e)} \ge c_\ell z'_e  \cdot \bm{1}_{\neg C_{\ell}(e)}.
\end{equation*}
Thus, after taking expectations over $(Y_f)_{f \in E \setminus e}$, 
\begin{equation} \label{eqn:semi_avg_psi_performance}
    \Pr[\{e \in \psi_{\ell}(G', \bm{z}')\} \cap \neg C_{\ell}(e) \mid Y_e =1] \ge c_\ell \mb{E}[z'_e \cdot \bm{1}_{\neg C_{\ell}(e)}\mid Y_e =1].
\end{equation}
We now must lower bound $\mb{E}[z'_e \cdot \bm{1}_{\neg C_{\ell}(e)}\mid Y_e =1]$.
By applying Lemmas \ref{lem:good_fractional_matching} and \ref{lem:cycle_upper_bound},
we get that
\begin{align*}
\mb{E}[z'_e \cdot \bm{1}_{\neg C_{\ell}(e)}\mid Y_e =1] &\ge \frac{1}{\log(1/\eps)} \left(1 - O( \eps ( \log(1/\eps))^{2\ell+1}) - 2\exp\left(- \Omega(\log^{1/2}(1/\eps))\right) \right) \\
&= (1 - o(1)) \frac{1}{\log(1/\eps)},
\end{align*}
where the final line uses that $\eps ( \log(1/\eps))^{2\ell+1}) = o(1)$ for
$\ell(n) = \log\log(1/\eps)$.
By combining everything together,
\begin{align*}
     \Pr[e \in \psi_{\ell}(G', \bm{z}') \mid Y_e =1] \ge \Pr[\{e \in \psi_{\ell}(G', \bm{z}')\} \cap \neg C_{\ell}(e) \mid Y_e =1]  \ge (1 - o(1)) \frac{c_\ell}{\log(1/\eps)}.
\end{align*}
Finally,
recalling that $y_e = (1 - 1/\log^{1/4}(\eps^{-1})) x_e \revcolor{\cdot} \log(1/\eps)$, we get
that
$$
\Pr[e \in \psi_{\ell}(G', \bm{z}') \revcolor{ \mid X_e = 1}] \ge (1 - o(1)) c_\ell = (1 - o(1)) c
$$
after cancellation, and using the fact that $c_\ell \rightarrow c$ as $n \rightarrow \infty$, since $\ell \rightarrow \infty$ as $n \rightarrow \infty$. Thus,
$\psi$ is $c$-selectable for graphs with vanishing edges values.

To complete the proof, we must argue that if each $\psi_{\ell}$ is strongly-online,
then $\psi$ can be implemented as an OCRS as the edge states $(X_e)_{e \in E}$ of $(G,\bm{x})$ are revealed. First observe that $\psi$ is given $(G,\bm{x})$ in advance, and thus knows $\eps$. This means that steps \ref{enum:subgraph_sample}. and \ref{enum:fractional_matching}. (i.e., \Cref{alg:fractional_matching}) of $\psi$ involving edge $e \in E$ can be implemented when
$e$ arrives. In order to execute step \ref{enum:execute_tree_local}. online, we pass on matching $e$ if $Y_e = 0$ (i.e., $e$ is not an edge of $G'$).
Otherwise, we make use of the strongly online property of $\psi_\ell$: Observe
that the values $(z'_f, Z'_f)_{f \le e, Y_f =1}$ are known when $e$ arrives, and determine
whether or not $e \in \psi_{\ell}(\bm{z}', G')$. Thus, when $Y_e =1$,
we can present $\psi_\ell$ the values $(z'_f, Z'_f)_{f \le e, Y_f =1}$ to determine whether to match $e$. The resulting
matching will be identical to the offline setting in which $\psi_\ell$ is given the entire input $(G', \bm{z}')$ in advance.
\end{proof}

\section{OCRS Negative Results} \label{sec:OCRS_neg}

\subsection{General Graphs} \label{sec:OCRS_neg_general}
We define a sequence of inputs dependent on $n \in \mb{N}$ to prove the \revcolor{upper bound}.
Start with a vertex $u$ whose neighbors $N(u)=\{v_1, \ldots, v_n\}$ have edge values $x_{u,v_i} = 1/n$ for $i \in [n]$. 
Next, add a clique amongst the vertices of $N(u)$
where $x_{v_i,v_j} =1/n(n-1)$ for each $i,j \in [n]$, $i \neq j$. Finally, each $v_i$ has its own $n-2$
additional neighbors, denoted $W_i$, where $x_{v_i,w} = 1/n$ for each $w \in W_i$ and $i =1, \ldots ,n$.
We define $W = \cup_{i=1}^{n} W_i$ for convenience, and refer to the input by $(G(n),\bm{x}(n))$,
which we denote by $(G,\bm{x})$ when clear (see \Cref{fig:negativeresult} for an illustration).
We next specify the order in which the edges are presented to the OCRS. This is done in $3$ \textit{phases}: 
\begin{enumerate}
\item Present the edges of $N(u) \times W$ in an arbitrary order. 
\item Present the edges between the vertices of $N(u)$ in an arbitrary order.
\item Present the edges of $\partial(u)$ in an arbitrary order.
\end{enumerate}
We refer to an order specified in this way as a \textit{phase-based order}. Our negative result holds
for any such order, even if it is presented to the OCRS ahead of time. Let us define  $\beta \in (0,1)$ to be the unique solution to the below equation:
\begin{equation}\label{eqn:beta_negative}
    0 =1 - 2 \beta- \beta \exp(2 -1/\beta),
\end{equation}
where $\beta \approx 0.3895$. 
\begin{theorem} \label{thm:negative_OCRS_input}
Any OCRS which processes the edges of $(G,\bm{x})$ in a phase-based order is at most $(\beta + o(1))$-selectable on $(G,\bm{x})$.

\end{theorem}
\begin{figure}[ht] 
\centering
\resizebox{0.4\textwidth}{!}{%
\begin{tikzpicture}[node style/.style={circle, draw, minimum size=0.5em, inner sep=0pt},
                    third layer node style/.style={circle, draw, minimum size=1.2em, inner sep=0pt}]
  \node[node style, label={[font=\tiny]above:$u$}] (root) at (0,0) {};
  
  \node[node style, label={[font=\tiny]left:$v_{1}$}] (child1) at (-2,-1.25) {};
  \node[node style, label={[font=\tiny]left:$v_{2}$}] (child2) at (-1,-2) {}; 
  \node[node style, label={[font=\tiny]left:$v_{3}$}] (child3) at (0,-2.25) {};
  \node[node style, label={[font=\tiny]right:$v_{4}$}] (child4) at (1,-2) {}; 
  \node[node style, label={[font=\tiny]right:$v_{n}$}] (child5) at (2,-1.25) {};

  \node[font=\tiny, rotate=45] at (1.5,-1.625) {$\cdots$};

  \node[third layer node style, label={[font=\tiny]below:$W_{1}$}] (grandchild1) at (-2,-2.75) {};
  \node[third layer node style, label={[font=\tiny]below:$W_{2}$}] (grandchild2) at (-1,-3.5) {}; 
  \node[third layer node style, label={[font=\tiny]below:$W_{3}$}] (grandchild3) at (0,-3.75) {};
  \node[third layer node style, label={[font=\tiny]below:$W_{4}$}] (grandchild4) at (1,-3.5) {}; 
  \node[third layer node style, label={[font=\tiny]below:$W_{n}$}] (grandchild5) at (2,-2.75) {};

  \foreach \i/\x/\y in {1/-2/-2.75, 2/-1/-3.5, 3/0/-3.75, 4/1/-3.5, n/2/-2.75}
  {
    \foreach \a in {0,120,240}
      \draw (\x,\y) +(\a:0.3em) circle (0.1em);
  }

  \node[font=\tiny, rotate=45] at (1.5,-3.125) {$\cdots$};

  \foreach \x in {1,...,5}
    \draw (child\x) -- (grandchild\x) node[midway, fill=white, font=\tiny] {$1-\frac{2}{n}$};


\begin{scope}[lightgray, opacity=0.5]
    \foreach \x in {1,...,5}
        \foreach \y in {\x,...,5}
            \draw (child\x) -- (child\y);.
\end{scope}
\path[lightgray] (child2) -- (child5) node[midway, fill=white, font=\tiny] {$\frac{1}{n(n-1)}$};

\foreach \x in {1,...,5}
    \draw (root) -- (child\x) node[midway, fill=white, font=\tiny] {$\frac{1}{n}$};

\node[node style] (child1) at (-2,-1.25) {};
\node[node style] (child2) at (-1,-2) {};
\node[node style] (child3) at (0,-2.25) {};
\node[node style] (child4) at (1,-2) {};
\node[node style] (child5) at (2,-1.25) {};

\end{tikzpicture}
}
\caption{Graph to prove negative result. Each $W_i$ consists of $n - 2$ vertices.}
\label{fig:negativeresult}
\end{figure}
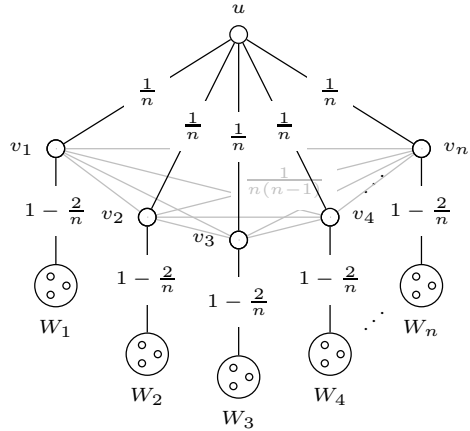
In order to prove \Cref{thm:negative_OCRS_input}, 
let us fix an OCRS for $(G,\bm{x})$ which
returns the matching $\scr{M}$ when presented the edges in a phase-based order,
and which is $c$-selectable for $c \in (0,1)$. Our goal is to upper bound $c$.
It is not hard to show that without loss of generality, 
we may assume that the OCRS is \textit{exactly}\footnote{Given a $c$-selectable OCRS for $(G,\bm{x})$ which returns the matching $\scr{M}$, we can turn it into an exactly $c$-selectable OCRS by independently \textit{dropping} any selected edge with probability $\Pr[e \in \scr{M} \mid X_e =1] -c$. {\color{blue} We then continue to simulate the original OCRS, including its dropped edges, so this reduction can be executed as the edges arrive online. In the concentrated case, we work directly with the original OCRS and do not use this reduction.}} $c$-selectable. That is, $\Pr[e \in \scr{M} \mid X_e =1] = c$ for each edge $e$ of $G$. 

Define $M_{v_i}$ (respectively, $A_{v_i} = 1 - M_{v_i}$) to be an indicator random variable for the event 
vertex $v_i$ is matched (respectively, not matched) after the first phase. 
Observe first that for each $i \in [n]$,
\begin{equation} \label{eqn:availability}
    \mb{E}[M_{v_i}] = \Pr[ \cup_{w \in W_i} (v_i,w) \in \scr{M}] = \sum_{w \in W_i} c x_{v, w} = c \left(1 - \frac{2}{n}\right),
\end{equation}     
where the last equality uses that $\sum_{w \in W_i} x_{v_i, w} =1 - 2/n$.
At this
point, let us consider the minimum covariance amongst the $(A_{v_i})_{i=1}^{n}$ random variables.
Specifically, define $\theta := \min_{i \neq j} \Cov[A_{v_i}, A_{v_j}]$, where $\Cov[A_{v_i},A_{v_j}] := \mb{E}[A_{v_i} A_{v_j}] - \mb{E}[A_{v_i}] \mb{E}[A_{v_j}]$.
We first lower bound $\theta$ as a function of $c$ for $n \rightarrow \infty$. This gives a tighter upper bound on $c$ the \textit{smaller} $\theta$ is.
\begin{restatable}[Small covariance regime: proven in \S\ref{pf:lem:small_variance}]{lemma}{smallVariance} \label{lem:small_variance}
$
    \theta \ge  (1 +o(1))(c - (1-c)^2).
$
\end{restatable}
\revcolor{This lemma follows from examining each of the edges $(v_i, v_j)$, and noting that we can only select it if each of its end points is available, so $c$ must be less than $\mb{E}[A_{v_i} A_{v_j}]$, which is less than $\mb{E}[A_{v_i}] \mb{E}[A_{v_j}] + \theta$.}
We next upper bound a function of $c$ and $\theta$ by $1$. This leads to a tighter
upper bound on $c$ the \textit{larger} $\theta$ is.
\begin{restatable}[Large covariance regime: proven in \S\ref{pf:lem:large_variance}]{lemma}{largeVariance} \label{lem:large_variance}
$
    2c + c \exp(c +\theta/c -1) \le 1 + o(1).
$   
\end{restatable}
\revcolor{The main idea behind this lemma is to examine the last edge, consider the probability that both its vertices are already matched when the edge arrives, and show that this probability is not too large. When the $A_{v_i}$ are positively correlated, conditioning on $A_{v_n}=0$ leaves fewer available neighbors for $u$ to get matched with. If none of the edges to these neighbors is active, there is no way $u$ could be matched before the arrival of the last edge. Considering this activity probability gives the exponential term after applying Jensen's inequality}

We are now ready to prove \Cref{thm:negative_OCRS_input}. We combine the inequalities from Lemmas \ref{lem:small_variance} and \ref{lem:large_variance} to determine the optimal choice of $\theta$ for the OCRS. 
\begin{proof}[Proof of \Cref{thm:negative_OCRS_input}]
Recall that $\beta$ satisfies \eqref{eqn:beta_negative}.
Now, for any OCRS which is exactly $c$-selectable on $(G, \bm{x})$,
Lemmas \ref{lem:small_variance} and \ref{lem:large_variance} together imply that
$$
     2c + c \exp(c +(c - (1-c)^2)/c -1) \le 1 + o(1).
$$
Thus, $1 - 2c - c \exp(2 - 1/c) \ge o(1)$ after rearranging and simplifying. Now,
the function $z \rightarrow 1 - 2z - z \exp(2 - 1/z)$ is continuous, and non-negative 
for $z \le \beta$, so we may conclude that $c \le \beta + o(1)$, as required.
\end{proof}

\subsection{Trees} \label{sec:OCRS_neg_trees}

In this section, we first establish a negative result for \textit{concentrated} OCRSs on trees. We refer to an OCRS as \textit{concentrated} on $(G,\bm{x})$, provided the matching $\scr{M}$ on $(G,\bm{x})$ it returns satisfies $||\scr{M}| - \mb{E}[ |\scr{M}|]| = o(\mb{E}[ |\scr{M}|)$ w.h.p. {\color{blue}As we discuss in \Cref{apx:concentration-property}, \Cref{alg:tree_aom}, the OCRS of \citet{Ezra_2020}, and the OCRS of \citet{aminian2026stationary} all satisfy this property.
}  Throughout this section, let us fix a particular concentrated OCRS for $(G,\bm{x})$ which
returns a matching $\scr{M}$, and is $c$-selectable for $c \in (0,1)$, i.e., $\Pr[e \in \scr{M} \mid X_e =1] = c$ for each edge $e$ of $G$.

As in the hardness result for general graphs, we begin by defining a sequence of inputs dependent on $n \in \mb{N}$ to prove the hardness result. Start with a vertex $u$ whose neighbors $N(u)=\{v_1, \ldots, v_n\}$ have edge values $x_{u,v_i} = 1/n$ for $i \in [n]$. Each $v_i$ has its own $n-1$ additional neighbors, denoted $W_i$, where $x_{v_i,w} = 1/n$ for each $w \in W_i$ and $i =1, \ldots ,n$. We refer to this input by $(G(n),\bm{x}(n))$ (see \Cref{fig:negativeresulttrees} for an illustration).
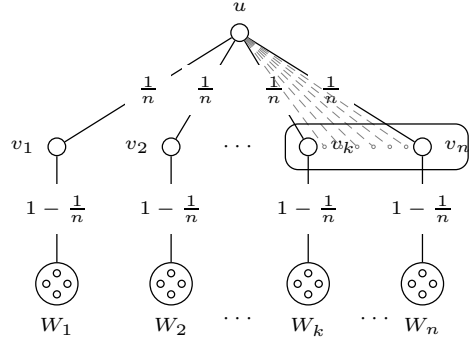
\begin{figure}[ht] 
\centering
\resizebox{0.4\textwidth}{!}{%
\begin{tikzpicture}[node style/.style={circle, draw, minimum size=0.5em, inner sep=0pt},
                    third layer node style/.style={circle, draw, minimum size=1.2em, inner sep=0pt},
                    faint edge/.style={very thin, gray, densely dashed}, 
                    faint node/.style={circle, draw, gray, very thin, minimum size=0.1em, inner sep=0pt}]
  \node[node style, label={[font=\tiny]above:$u$}] (root) at (0,0) {};
  
  \node[node style, label={[font=\tiny]left:$v_{1}$}] (child1) at (-2,-1.25) {};
  \node[node style, label={[font=\tiny]left:$v_{2}$}] (child2) at (-0.75,-1.25) {};
  \node[node style, label={[font=\tiny]right:$v_{k}$}] (child4) at (0.75,-1.25) {};
  \node[node style, label={[font=\tiny]right:$v_{n}$}] (child5) at (2,-1.25) {};

\node[font=\tiny] at (0,-1.25) {$\cdots$};

  \node[third layer node style, label={[font=\tiny]below:$W_{1}$}] (grandchild1) at (-2,-2.75) {};
  \node[third layer node style, label={[font=\tiny]below:$W_{2}$}] (grandchild2) at (-0.75,-2.75) {};
  \node[third layer node style, label={[font=\tiny]below:$W_{k}$}] (grandchild4) at (0.75,-2.75) {};
  \node[third layer node style, label={[font=\tiny]below:$W_{n}$}] (grandchild5) at (2,-2.75) {};

    \node[font=\tiny] at (0,-3.125) {$\cdots$};
    \node[font=\tiny] at (1.5,-3.125) {$\cdots$};

  \foreach \x/\y in {1/grandchild1, 2/grandchild2, 4/grandchild4, 5/grandchild5}{
    \draw (child\x) -- (\y) node[midway, fill=white, font=\tiny] {$1-\frac{1}{n}$};
  }
  \foreach \x in {1,...,2,4,5}{
    \draw (root) -- (child\x) node[midway, fill=white, font=\tiny] {$\frac{1}{n}$};
  }

   \foreach \i/\x/\y in {1/-2/-2.75, 2/-0.75/-2.75, 4/0.75/-2.75, n/2/-2.75}
  {
    \foreach \a in {0,90,180, 270}
      \draw (\x,\y) +(\a:0.3em) circle (0.1em);
  }

\foreach \x in {1,...,6}{
  \pgfmathsetmacro{\xpos}{0.75 + (2 - 0.75)*\x/7} 
  \node[faint node] (faintnode\x) at (\xpos, -1.25) {}; 
  \draw[faint edge] (root) -- (faintnode\x);
}

 \draw[rounded corners] (0.5,-1) rectangle (2.5,-1.5);
\end{tikzpicture}
}
\caption{Tree to prove negative result. Each $W_i$ consists of $n - 1$ vertices.}
\label{fig:negativeresulttrees}
\end{figure}
The edges are presented to the OCRS in $2$ phases. In the first phase, the edges of $N(u) \times \cup_{i=1}^{n} W_i$ are presented in an arbitrary order. At this point, we define $M_{v_i}$ (respectively, $A_{v_i} = 1 - M_{v_i}$) as an indicator random variable for the event that vertex $v_i$ is matched (respectively, not matched) after the first phase. \revcolor{The distribution of $(A_{v_i})_{i=1}^n$ is not affected by the arrival order of the second-phase edges.} Before we can discuss the order in which the edges of the second phase are presented to the OCRS, we must briefly digress to discuss an important property of a concentrated OCRS.

Since $\sum_{i=1}^{n} M_{v_i}$ and $|\scr{M}|$ differ by at most 1 and we've assumed that $|\scr{M}|$ is concentrated about $\mb{E}[|\scr{M}|]$, we know that $\sum_{i=1}^{n} M_{v_i}$ is also concentrated. Thus, by applying elementary bounds, we get that $\Var[\sum_{i=1}^{n} M_{v_i}] = o(n^2)$, and so we must have that $\Var[\sum_{i=1}^{n} A_{v_i}] = o(n^2)$.  
 For concreteness, let us say that $\Var[\sum_{i=1}^{n} A_{v_i}] \leq \theta_nn^2$, with $\theta_n \to 0$, and $\theta_n \geq \frac{1}{\sqrt{n}}$. Note that since we are only demanding an upper bound on $\Var[\sum_{i=1}^{n} A_{v_i}]$, we can assume w.l.o.g. $\theta_n \geq \frac{1}{\sqrt{n}}$. This will just be a convenient upper bound to use on the `true' $\theta_n$ in the following calculations.

We have the following probabilistic lemma: 
\begin{restatable}[proven in \S\ref{pf:existence_of_end}]{lemma}{existenceOfEnd} 
    \label{existence_of_end}
Suppose $n \to \infty$. Let $A_1, A_2, \ldots, A_n$ be random variables, with $\Var(A_i) < \delta$. Define $A = \sum_{i=1}^n A_i$, and $A_S = \sum_{i \in S}A_i$. Suppose $\text{Var} (A) \leq  \theta_nn^2$ and $\theta_n \geq \frac{1}{\sqrt{n}}$. Then, we can find $S \subseteq \{1, 2, \ldots, n\}$ with $|S| = m = \theta_n^{1/4}n$ such that $\text{Var} (A_S) \leq 3\theta_nm^2$ for sufficiently large $n$.  
\end{restatable}

{\color{blue} We prove this lemma simply by considering the average value of the variance over all possible sets $S$ of size $m$.} The lemma implies that for $m = \theta_n^{1/4}n$, we can find a set $S \subseteq \{1, 2, \ldots, n\}$ of size $m$ such that $\Var (\sum_{i \in S}A_{v_i}) \leq 3\theta_nm^2$. Let us assume that in phase 2, the edges of $\partial(u)$ are presented in an order that ensures that the edges $\{u\} \times S = \{(u, v_i): v_i \in S\}$ come last. We emphasize that this is a valid ordering,
as the adversary has access to the OCRSs description ahead of time.

\begin{theorem} \label{thm:negative_OCRS_trees_concentrated}
Any concentrated OCRS which processes the edges of $(G,\bm{x})$ in two phases as described above is at most $(\alpha + o(1))$-selectable, where $\alpha$ is defined in \eqref{eqn:alpha}.
\end{theorem}

For convenience, let us immediately relabel the vertices so that the edges of $\partial(u)$ are presented in the order $e_1 = (u,v_1), \ldots , e_n =  (u,v_n)$ to the OCRS. We will \revcolor{group} the last $m$ vertices from $v_k$ to $v_n$ together ($k = n- m+1$), and think about them as if they were {\color{blue} one last vertex. We do this to ensure that the algorithm's decisions before reaching the last vertex cannot be very correlated with the availability of the last vertex. } It will be useful to define:
\[A_{end}= \sum_{i=k}^n A_{v_i}, \mathcal{B} = \{\cup_{j =k}^{n} \{ X_{e_j}A_{v_j} =1\}\}, \text{ and } \mathcal{C} = \{M_u(e_k) = 0\}.\]
Here $\mathcal{B}$ represents the event that there is something for the OCRS to select amongst the last $m$ edges, and $\mathcal{C}$ represents the event that the OCRS even reaches the last $m$ vertices, i.e., it has not already made a selection when it reaches $e_k$. Note that we have ensured that $\Var(A_{end}) = o(m^2)$.

Since the OCRS must select each of the last $m$ edges $e_k, \ldots, e_n$ with probability {\color{blue}at least} $\frac{c}{n}$, and it can only do this if both the events $\mathcal{B}$ and $\mathcal{C}$ happen, we conclude that:
\begin{align*}
\frac{cm}{n}\leq \Pr[\mathcal{B} \cap \mathcal{C}],
\end{align*} 
{\color{blue} We would now like to prove an upper bound on the probability that $\mathcal{B}$ and $\mathcal{C}$ simultaneously occur. Since $A_{end}$ is concentrated, it is approximately deterministic. Therefore, the event $\mathcal{B}$, which is a function of $A_{end}$ and $X_{e_j}$ for $j \geq k$, must necessarily be independent of the algorithm's matching decisions before $e_k$ since these decisions are made without knowledge of $X_{e_j}$ for $j \geq k$. This observation lets us establish following lemma:}

\begin{restatable}[proven in \S\ref{pf:lem:concentrated_ocrs_analysis}]{lemma}{concentratedOCRSAnalysis} \label{lem:concentrated_ocrs_analysis}
{\color{blue} Suppose $\Pr[A_{end} > T] \leq \eps$. Then,
\[\Pr[\mathcal{B} \cap \mathcal{C}] \leq \frac{T}{n}\Pr[\mathcal{C}]+\eps \leq \frac{T}{n}\left(1-\frac{c(n-m)}{n}\right)+\eps.\]}
\end{restatable}
We are now ready to prove \Cref{thm:negative_OCRS_trees_concentrated}.
\begin{proof}[Proof of \Cref{thm:negative_OCRS_trees_concentrated}]
Note that we assumed $m = \theta_n^{1/4}n$. Plugging this into the inequality, and using \Cref{lem:concentrated_ocrs_analysis} we get that
\begin{align*}
c\theta_n^{1/4} \leq {\color{blue}\frac{T}{n}\left(1-c+c\theta_n^{1/4}\right)+\eps}.
\end{align*}
{\color{blue}Since each $v_i$ is matched in the first phase with probability at least $c(n-1)/n$, we have $\E[A_{end}]\leq (1-c+c/n)m$.}
If we let $T = (1-c{\color{blue}+c/n})m + \theta_n^{1/4}m$, Chebyshev's inequality tells us that $\Pr[A_{end} > T]\leq \frac{3\theta_nm^2}{\theta_n^{1/2}m^2} = 3\theta_n^{1/2}.$ Therefore,
\begin{eqnarray*}
&&c \leq {\color{blue}\left(1-c+c/n+\theta_n^{1/4}\right)\left(1-c+c\theta_n^{1/4}\right)+3\theta_n^{1/4}}.
\end{eqnarray*}
Taking $n \to \infty$ yields the desired result.
\end{proof}
We next consider when the OCRS is not necessarily concentrated on $(G,\bm{x})$. The proof here follows
closely to the proof of \Cref{thm:negative_OCRS_input}, with the main difference being that we focus on the average covariance of $(A_{v_i})_{i=1}^n$, as opposed to the minimum covariance. 
\begin{restatable}[proven in \S\ref{pf:thm:negative_OCRS_trees}]{theorem}{treeUnconditional}  \label{thm:negative_OCRS_trees}
Any OCRS is at most $(\gamma + o(1))$-selectable on $(G,\bm{x})$ where $\gamma$ is the unique solution to the equation: 
$0 = 1 - 2 \gamma- \gamma\exp(\gamma -1),$
and $\gamma \approx 0.3929$.
\end{restatable}

\section{RO-CRS and FO-CRS Positive Results} \label{sec:rcrs_pos}

We now investigate contention resolution schemes that assume that the edges of the input $(G,\bm{x}) =(G(n), \bm{x}(n))$ arrive in a random order drawn u.a.r., or in an order chosen by the CRS. In the latter case, the edge-ordering is a \textit{non-uniform} random ordering. For both arrival orders, we study the \textsc{greedy-CRS} which matches an arriving edge $e=(u,v)$ provided $e$ is active,
and $u$ and $v$ were previously unmatched.

For general inputs with vanishing edge values, \citet{brubach2021offline} showed
that \textsc{greedy-CRS} is $(1-e^{-2})/2 \approx 0.432$-selectable,
and it is not hard to see that their analysis is tight\footnote{Consider $G=(V,E)$ with $V=\{u_i,v_i\}_{i=0}^n$ and $E=\{(u_0,v_0)\}\cup\{(u_0,u_i),(v_0,v_i)\}_{i=1}^n$ The probability $x_e$ of every $e\in E$ is $1/(n+1)$. When $(u_0,v_0)$ is active, \textsc{greedy-CRS} selects $(u_0,v_0)$ with probability \textit{exactly} $(1-e^{-2})/2$ as $n\to\infty$.}. In order to surpass this bound, one approach is to
ensure that $(G,\bm{x})$ is $1$-\textit{regular}; that is, $\sum_{e \in \partial(v)} x_e =1$
for each $v \in V$. \citet{Fu2021} were the first to prove a reduction to this setting in the related vertex arrival model \revcolor{(see Lemma 2 in Section 3.1)}. \citet{macrury2023random} applied this reduction for random-order edge arrivals \revcolor{(Lemma 1 in Appendix C)}, and \citet{nuti2023towards} applied it in the offline setting. \citet{nuti2023towards} also observed that this reduction preserves vanishing edge values, no matter the arrival order of the edges.
We therefore make this assumption throughout the section. 
\begin{theorem} \label{thm:ROCRS_positive}
For $1$-regular inputs with vanishing edge values, \textsc{greedy-CRS} is $\frac{1}{2}$-selectable when the edges arrive u.a.r.
\end{theorem}

Recall that it is impossible for an RCRS to be better than $\frac{1}{2}$-selectable even for vanishing edge values \revcolor{(see Theorem 5 in Section 3.3 in \citet{macrury2023random})}.
We next show that it is possible to surpass this impossibility result in the free-order setting, where we now analyze \textsc{Greedy-CRS} in a \textit{non-uniform} random order.

\begin{theorem} \label{thm:FOCRS_positive}
For $1$-regular inputs with vanishing edge values, \textsc{greedy-CRS} is $1-\ln\left(2-\frac{1}{e}\right) \approx 0.510$-selectable when the edges arrive in a certain non-uniform random order.
\end{theorem}

It will be convenient to assume for the proof of Theorem~\ref{thm:ROCRS_positive} that each edge $e$ arrives at a time that is chosen independently and u.a.r. from $[0, 1]$. Furthermore, for the proof of Theorem~\ref{thm:FOCRS_positive}, we will assume that we generate for each vertex an independent seed u.a.r. from $[0, 1]$, and that the edges then arrive in lexicographic order, i.e., each edge is associated with the seeds of its vertices $\{x,y\}$, and it arrives at time $\min\{x,y\} + \epsilon \max\{x,y\}$, where $\epsilon$ is any number less than the distance between any two vertex seeds. The primary property we will need for our ordering is that an edge $(u, w)$ is processed before $(u, v)$ exactly when $w$'s seed is less than $v$'s seed, so our proof will also work for orderings other than the lexicographic one. \revcolor{In any such case, we note that our analysis of \textsc{Greedy-CRS} is in fact tight.} The proofs of Theorems~\ref{thm:ROCRS_positive} and \ref{thm:FOCRS_positive} both first reduce to the random tree case, and so we treat them in a unified manner.

\subsection{Reduction to the Random Tree Case}
This section follows \citet{nuti2023towards} very closely; however for completeness
we present all the proofs in \Cref{sec:RCRS_additions}.
\revcolor{Recall that $S_e$ is the connected component of $e$ induced by the active edges of $G$.} Consider a simple method to construct a random tree $\T$ that will estimate $S_e$ as $n \rightarrow \infty$. Let $u$ and $v$ be two special vertices connected by an edge, and generate the random tree $\T$ by generating two independent Galton-Watson processes at $u$ and $v$. Explicitly, construct $\T$ as follows: Start with the two special vertices $u$ and $v$ unmarked and connected by an edge, and add them into $\T$. Then, for each unmarked vertex $w$ added into $\T$, obtain an independent sample $L_w$ of a Poisson random variable with mean 1, add $L_w$ new children to $w$, and then mark $w$. End the construction when all vertices of $\T$ are marked. Standard results about Galton-Watson processes immediately tell us that $\T$ is finite with probability $1$ (e.g., see Theorem 6.1 in \citet{branchingprocessestext}).

After we generate $\T$, clearly it is possible to run a greedy matching algorithm on $\T$ (either after generating random times of arrival for each edge $e \in \T$, or random seeds for each vertex $w \in \T$ as the case may be). This produces a matching $\scr{M}_{\T}$. We claim the following relation between the execution of the greedy CRS on $G(n)$ and the greedy matching algorithm on $\T$:

\begin{restatable}[proven in \S\ref{sec:RCRS_additions}]{theorem}{reductiontogw}  \label{thm:reductiontogw}
There exists a function $f$ satisfying $\lim_{\eps \to 0}f(\eps) = 0$ such that for any $1$-regular input $(G(n), \bm{x}(n))$ with vanishing edge values, and for all $e \in G(n)$, we have:
$$|\Pr[e\in \scr{M}\mid X_e = 1] - \Pr[(u, v) \in  \scr{M}_{\T}]| \leq f(\eps(n))\revcolor{,}$$
\revcolor{where $\scr{M}$ is the matching produced by \textsc{greedy-CRS} when run on $G(n)$.} 
\end{restatable}
Clearly, establishing \Cref{thm:reductiontogw} reduces the proofs of Theorems \ref{thm:ROCRS_positive} and \ref{thm:FOCRS_positive} to the calculation of $\Pr[(u, v) \in \scr{M}_{\T}]$ in two different random orderings of the edges in $\T$. 

\subsection{Proving Theorems~\ref{thm:ROCRS_positive} and \ref{thm:FOCRS_positive}}
Before we begin the proofs of these theorems, we will establish some common notions between the two proofs. First, we say that a vertex $w \in \T$ has been matched if an edge adjacent to it has been matched by the greedy matching algorithm.
Second, in order for $u$ to be unmatched when $(u, v)$ arrives, we need every edge $(u, w)$ from $u$ to one of its children $w$ to satisfy:
\begin{enumerate}
    \item It is processed after edge $(u,v)$, \textit{or}
    \item $w$ is already matched when $(u,w)$ is processed.
\end{enumerate}
Let us say that a child $w$ of $u$ is \textit{bad} if it satisfies neither of these criteria, i.e., $(u,w)$ is processed before $(u,v)$ and $w$ is unmatched at this point. Observe that $(u, v)$ is matched by the greedy algorithm precisely when no child of $u$ or $v$ is bad.

\begin{proof}[Proof of \Cref{thm:ROCRS_positive}]
We aim to calculate $\Pr[(u, v) \in \scr{M}_{\T}]$ when each edge $e$ of the random tree $\T$ arrives at a time that is chosen independently and uniformly at random from $[0, 1]$. 

Let us define for any $x \in [0,1]$, the function $\q(x)$ which denotes the probability that the special vertex $u$ is \textit{unmatched} at the time $x$ when $(u, v)$ arrives. The probability is taken over all possible realizations of $\T$, and all possible arrival times of the edges in $\T$ (other than $(u, v)$).  Note briefly that $q$ is continuous---indeed, \[|\q(x)-q(x+h)| \leq \Pr[\text{some edge arrives in between times } x \text{ and } x + h].\]

The probability that $(u,v)$ is selected by the greedy algorithm is
$\int_0^1 \q(x)^2 dx.$
This is because:
\begin{enumerate}
    \item The symmetry between $u$ and $v$ ensures $\q(x)$ is also the conditional probability that $v$ is unmatched.
    \item Whether $u$ or $v$ are matched before the edge $(u, v)$ arrives are independent events.
    \item $(u, v)$ is added into the matching exactly if neither $u$ nor $v$ are matched when it arrives.
\end{enumerate}
It turns out that we can actually determine what $\q(x)$ is using a recursion.  Critically, for an edge $(u,w)$ arriving at time $z$, the probability (taken over all possible realizations of how the process looks below $w$) that $w$ is unmatched when $(u,w)$ arrives is exactly $\q(z)$, and so, the probability that $w$ is bad is exactly $\int_0^x\q(z) dz$. \revcolor{Furthermore, the chance that $u$ has $k$ children is $\frac{e^{-1}}{k!}$, and whether these children are bad are independent events.} It follows that 
\begin{eqnarray*}
\q(x) = \sum_{k = 0}^{\infty}\Pr[u\text{ has $k$ children and none of them are bad}]
&= & \sum_{k = 0}^{\infty}\frac{e^{-1}}{k!}\left(1-\int_0^x\q(z) dz\right)^k\\
&= & \exp{\left(-\int_0^x\q(z) dz\right)}.\\
\end{eqnarray*}
To solve this equation, note that $q$ is differentiable in $x$ (since $\int_0^x\q(z) dz$ is differentiable in $x$) and $\frac{dq}{dx}= -\q(x)^2$. We know that $\q(0) = 1$, so this is a differential equation we can solve using separation of variables. We conclude that $\q(x) = \frac{1}{x+1}.$
Finally, note that 
\begin{eqnarray*}
\int_0^1 \q(x)^2 dx& = &-\int_0^1 \frac{d q}{dx} dx =  \q(0)-\q(1)=  \frac{1}{2}.
\end{eqnarray*}
This concludes the proof of Theorem~\ref{thm:ROCRS_positive}.
\end{proof}

The argument to prove Theorem~\ref{thm:FOCRS_positive} is similar to the previous proof, but we need to be slightly more careful in how we define $\q$. 

\begin{proof}[Proof of \Cref{thm:FOCRS_positive}]
 Once more, we aim to calculate $\Pr[(u, v) \in \scr{M}_{\T}],$ but this time each edge $e$ of the random tree $\T$ arrives at a time determined by the random seeds of the edge's associated vertices. 

Let us define for any $x, y \in [0,1]$, the function $\q(x, y)$ which denotes the probability that the special vertex $u$ is unmatched when $(u, v)$ arrives given $u$ has the seed $x$ and $v$ has the seed $y$. The probability that $(u,v)$ is selected by the greedy algorithm is now:
$$\int_0^1 \int_0^1 \q(x,y)\q(y,x) dx dy.$$ 
Once again, it turns out that we can establish a recurrence for $q$. Indeed, for an edge $(u,w)$ in which $w$ has seed $z$, the probability (taken over all possible realizations of how the process looks below $w$) that $w$ is unmatched when $(u,w)$ arrives is exactly $\q(z, x)$, and so, the probability that $(u, w)$ is bad is exactly $\int_0^y\q(z, x) dz$ (since the edge $(u, w)$ is processed before $(u, v)$ exactly if $z \leq y$).
It follows that 
\begin{eqnarray*}
\q(x, y)& = &\sum_{k = 0}^{\infty}\Pr[u\text{ has } k \text{ children and none of them are bad}]\\
&= & \sum_{k = 0}^{\infty}\frac{e^{-1}}{k!}\left(1-\int_0^y\q(z, x) dz\right)^k\\
&= & \exp{\left(-\int_0^y\q(z, x) dz\right)}.\\
\end{eqnarray*}
To solve this equation, note that $q$ is differentiable in $y$ (since $\int_0^y\q(z, x) dz$ is differentiable in $y$) and 
$$\revcolor{\frac{\partial}{\partial y}\q(x, y)}= -\q(x, y)\q(y, x).$$ Since $\q(x, y)$ is differentiable in $y$,  $\int_0^y\q(z, x) dz$ is differentiable in $x$, and so $\q$ is differentiable in $x$ as well. Then, 
\begin{eqnarray*}
\revcolor{\frac{\partial}{\partial x}\q(x,y)}& = &\q(x,y)\left(-\int_0^y \revcolor{\frac{\partial}{\partial x}\q(z, x)} dz \right)\\
&= & \q(x,y)\left(\int_0^y \q(z, x)\q(x, z) dz \right)\\
&= & \q(x,y)\left(-\int_0^y \revcolor{\frac{\partial}{\partial z}\q(x, z)} dz \right)\\
&=&\q(x, y)\left(1-\q(x,y)\right),
\end{eqnarray*}
since $\q(x, 0) = 1$. We know that $\q(0, y) = e^{-y}$, so this is a differential equation we can solve for each value of $y$ separately. It follows using variable separation that 
$$\q(x, y) = \frac{e^x}{e^x+e^y-1}.$$
Finally, note that 
\begin{eqnarray*}
\int_0^1 \int_0^1 \q(x,y)\q(y,x) dx dy& = &-\int_0^1\int_0^1 \revcolor{\frac{\partial}{\partial y}\q(x,y)} dy dx\\
&= & \int_0^1 \q(x, 0)-\q(x, 1)dx\\
&= & 1-\int_0^1 \frac{e^x}{e^x+e-1}dx\\
&=&1-\ln \left(2-\frac{1}{e}\right),
\end{eqnarray*}
where we computed the last integral by substituting $u = e^x$.
\end{proof}

\section{{\color{blue}FO-CRS Negative Result}} \label{sec:free_order_negative}


To prove \Cref{thm:fo-upper-bound},
we consider the complete bipartite graph with
$2 n$ vertices and uniform edge values. 
Specifically, let $G=(U_1,U_2, U_1 \times U_2)$ where $|U_1|= |U_2| =n$ for $n \ge 1$,
and set $x_e = p = 1/n$ for all $e \in U_1 \times U_2$. We denote this input by $(G, \bm{x})$, implicitly dropping the dependence on $n$. For $U \subseteq U_1 \cup U_2$, we denote $G[U]$ as the subgraph induced by vertices $U$. Furthermore, for a subgraph $H \subseteq G$, let $V(H)$ denote its vertices, and $E(H)$ denote its edges. Finally,
for each $F \subseteq E(H)$ and $U \subseteq V(H)$, let $H \setminus F$ be $H$ with the edges of $F$ removed, and $H - U$ be $H$ with the vertices of $U$ removed, including all of their incident edges in $H$. Throughout the section, all of our asymptotics are with respect to $n \rightarrow \infty$.

\begin{lemma} \label{lem:core_fo_upper_bound}
    For any FO-CRS which outputs matching $\scr M$ on $(G, \bm{x})$, $$\frac{\mb{E}[|\scr{M}|]}{n} \le (1 + o(1)) \left(1 - \ln \left(2 - \frac{1}{e}\right) \right).$$
\end{lemma}

\Cref{lem:core_fo_upper_bound} immediately implies \Cref{thm:free_order_positive}, so we focus on establishing it for the rest of the section. We note that any FO-CRS which maximizes $\mb{E}[|\scr M|]$
can be assumed to be deterministic without loss (by averaging over its internal random bits), and so we hereby make this assumption.

For each $1 \le t \le n^2$, let $e_t$ be the $t^{th}$ edge of $G$ processed by the FO-CRS,
and denote its state by $X_{e_t}$ (clearly, $X_{e_t} \sim \Ber(1/n)$). Let $\scr{M}_t$ denote
the matching constructed by the FO-CRS after $t$ rounds, and $F_t$ denote the \textit{free vertices} of $G$ which are unmatched by $\scr{M}_t$.
Since the FO-CRS is deterministic,
$\scr{M}_t$ and $F_t$ are functions of $(e_i, X_{e_i})_{i=1}^{t}$, as is the choice of $e_{t+1}$, the $(t+1)^{th}$ edge chosen by the FO-CRS. Thus, $\scr{M}_t$, $F_t$ and $e_{t+1}$
are measurable with respect to $\scr{H}_t$, the sigma-algebra generated from $(e_i, X_{e_i})_{i=1}^{t}$
(here $\scr{H}_{0}:= \{ \emptyset, \Omega\}$, the trivial sigma-algebra).
We refer to $\scr{H}_t$ as the \textit{history} after $t$ steps.
It will also be convenient to define the \textit{residual graph} $R_t$ after $t$ steps
by
\begin{equation} \label{eqn:residual_graph}
    R_t = G[F_t]  \setminus \{e_1, \ldots ,e_t\}.
\end{equation}
Notice that $R_t$ contains precisely those edges of $G$ which can feasibly added to the matching $\scr{M}_t$ after $e_1, \ldots , e_t$ have been processed.

Ideally, we would next upper bound $\mb{E}[ |\scr{M}_{t+1}| - |\scr{M}_t| \mid \scr{H}_t]$ for each
$t=0, \ldots , n^2-1$ in terms of $t$ and $|\scr{M}_t|$, and apply the one-sided differential equation method of \citet{bennett2023} to conclude that $|\scr{M}_t|/n$ is dominated by the solution of a differential equation. This is the approach taken in \citet{macrury2023random} for RCRS, where the authors use that since $e_{t+1}$ is distributed uniformly at random amongst $G \setminus \{e_1, \ldots , e_t\}$, 
$\mb{E}[ |\scr{M}_{t+1}| - |\scr{M}_t| \mid \scr{H}_t]$ is upper bounded by roughly $\frac{1}{n} \left( 1 - \frac{|\scr{M}_t|}{n} \right)^2$, and so $|\scr{M}_{n^2}|/n$ is at most $w(1) =1/2$, where $w$
is the unique solution to the differential equation $w'(\tau) = (1 - w(\tau))^2$ with $w(0) = 0$.

However, since $e_{t+1}$ is chosen adaptively by the FO-CRS, there is no such simple upper bound. Instead,
our approach is to define a \textit{surplus function} $\scr{S}$, which assigns a non-negative value $\scr{S}(R_t) \ge 0$ to the residual subgraph $R_t$.  Roughly speaking, $\scr{S}(R_t)$ approximately upper bounds the expected matching size available to an FO-CRS which begins with the subgraph $R_t$. Given such a function $\scr{S}$,
we consider the random variable
\begin{equation} \label{eqn:supermartingale}
    Q_t := |\scr{M}_t| + \scr{S}(R_t),
\end{equation}
and argue that for an appropriately chosen $\scr{S}$, $(Q_t)_{t=0}^{n^2}$ forms a supermartingale for which $Q_0 = \scr{S}(R_0) = \scr{S}(G)$
is our target upper bound of $(1 + o(1)) \left(1 - \ln \left(2 - \frac{1}{e}\right) \right)n$. Assuming
such a choice of $\scr{S}$ is possible, \Cref{lem:core_fo_upper_bound} follows easily, as we show in \Cref{lem:assuming_residual}.

In order to get a supermartingale, we need to ensure that when $e_{t+1}$ is processed, 
\begin{equation} \label{eqn:key_averaged_change}
\mb{E}[Q_{t+1} - Q_t \mid \scr{H}_t] \le 0.
\end{equation}
There are two main cases we must consider for \eqref{eqn:key_averaged_change} to hold.
If $e_{t+1}$ is rejected by the FO-CRS, then $|\scr{M}_{t+1}| = |\scr{M}_t|$, and yet $R_{t+1} = R_{t} \setminus e_{t+1}$. Thus,
we require $\mathcal{S}(R_t \setminus e_{t+1}) \le \mathcal{S}(R_t)$. On the other hand, suppose the FO-CRS accepts $e_{t+1}$
when it is active. Then,  if $X_{e_{t+1}} =1$, we have that $R_{t+1} = R_t - u_{t+1} - v_{t+1}$
and $|\scr{M}_{t+1}| = |\scr{M}_t| + 1$. Alternatively, if $X_{e_{t+1}} = 0$,
then $R_{t+1} = R_t \setminus e_{t+1}$ and $|\scr{M}_{t+1}| = |\scr{M}_t|$. By averaging over $X_{e_{t+1}} \in \{0,1\}$, this leads to the inequality $(1-p)\mathcal{S}(R_t \setminus e_{t+1})
    +p\bigl(1+\mathcal{S}(R_t - u_{t+1} - v_{t+1})\bigr)
    \leq \mathcal{S}(R_t)$. 

Given this motivation, we formalize the properties of a surplus function in the below definition:
\begin{definition} \label{def:residual_function}
The domain of a \textit{surplus function} $\scr{S}$ 
is the set of all subgraphs $H \subseteq G$, where $\scr{S}(H) \ge 0$ denotes the non-negative value assigned to $H \subseteq G$. Moreover,
for each $H \subseteq G$ and $(u,v) \in E(H)$, 
\begin{eqnarray} 
    &\mathcal{S}(H \setminus (u,v)) \le \mathcal{S}(H) \label{eq:fo_query_reject}\\
    &(1-p)\mathcal{S}(H \setminus (u,v))
    +p\bigl(1+\mathcal{S}(H - u - v)\bigr)
    \leq \mathcal{S}(H). \label{eq:fo_query_accept_or_reject}
\end{eqnarray}
Finally,
\begin{equation} \label{eq:fo_initial_value}
    \scr{S}(G) \le (1 + o(1))\left(1 - \ln \left(2 - \frac{1}{e}\right) \right)n.
\end{equation}
\end{definition}

\begin{restatable}[proven in \S\ref{pf:lem:assuming_residual}]{lemma}{assumingResidual}\label{lem:assuming_residual}
Suppose that $\scr{S}$ is a surplus function. Then, $(Q_t)_{t=1}^{n^2}$ as defined in \eqref{eqn:supermartingale} is a supermartingale with respect to the filtration $(\scr{F}_t)_{t=1}^{n^2}$. Moreover, $\mb{E}[ |\scr{M}_{n^2}|] \le (1+ o(1))\left(1 - \ln \left(2 - \frac{1}{e}\right) \right)n$.
\end{restatable} 

\subsection{Designing a Surplus Function} \label{sec:designing_residual_function}

We now construct a surplus function satisfying \Cref{def:residual_function}.
Since we are interested in the limit as $n\to\infty$, we may assume that
$n\geq 2$, and hence that $p=1/n\in(0,1)$.  Write $q:=1-p$.  Let $(B_v)_{v\in V(G)}$ be mutually independent random
variables with $B_v\sim\Ber(p)$.  We say that $v \in V$ is \textit{marked}, provided $B_v =1$. For each subgraph $H\subseteq G$, define $\mathcal{I}_H$ to be the event that the marked vertices form an independent
set in $H$. Formally,
\begin{equation} \label{eq:independent_set_event}
    \mathcal{I}_H
    :=
    \bigcap_{(u,v)\in E(H)}\{B_uB_v=0\}.
\end{equation}
Since $p < 1$, it is clear that 
$0<\Pr[\mathcal{I}_H]\leq 1$.
We next define the following function whose input in a subgraph $H \subseteq G$:
\begin{equation} \label{eq:residual_function_construction}
    \scr{S}(H):=-n\ln\Pr[\mathcal{I}_H].
\end{equation}
\begin{lemma} \label{lem:constructed_residual_function}
The function $\scr{S}$ defined in \eqref{eq:residual_function_construction} is
a surplus function in the sense of \Cref{def:residual_function}.
\end{lemma}
In order to motivate $\scr{S}$, it is helpful to first see
that it attains the desired value on $G$ as in \Cref{def:residual_function}.  Recall that
$G$ is the complete bipartite graph on $2n$ vertices. A marked set is independent in $G$ if and only if at least one
of the two parts $U_1,U_2$ contains no marked vertex.  Inclusion--exclusion
therefore gives
\begin{equation} \label{eq:independent_set_complete_bipartite}
    \Pr[\mathcal{I}_G]
    =
    2(1-p)^n-(1-p)^{2n}.
\end{equation}
Since $p=1/n$, it follows that
\begin{align}
    \scr{S}(G)
    &=
    -n\ln\left(2(1-1/n)^n-(1-1/n)^{2n}\right) \nonumber\\
    &=
    \left(
        1-\ln\left(2-\frac{1}{\mathrm e}\right)+o(1)
    \right)n. \label{eq:residual_function_initial_value}
\end{align}
It remains to verify the two remaining conditions in
\Cref{def:residual_function}.  We focus on \eqref{eq:fo_query_accept_or_reject}, since
it is harder to prove.  Substituting
\eqref{eq:residual_function_construction} into that inequality, dividing by
$n$, and using $p/n=p^2$ shows that it is equivalent to
\[
    \ln\Pr[\mathcal{I}_H]
    \leq
    q\ln\Pr[\mathcal{I}_{H\setminus\{e\}}]
    +p\ln\Pr[\mathcal{I}_{H-u-v}]
    -p^2.
\]
Exponentiating, this is
\begin{equation} \label{eq:probability_inequality_sufficient_for_residual}
    \Pr[\mathcal{I}_H]
    \leq
    \mathrm e^{-p^2}
    \Pr[\mathcal{I}_{H\setminus\{e\}}]^{q}
    \Pr[\mathcal{I}_{H-u-v}]^{p}.
\end{equation}
The next lemma proves the slightly stronger inequality obtained by replacing
$\mathrm e^{-p^2}$ with $1-p^2$.

\begin{restatable}[proven in \S\ref{pf:lem:independent_set_extension}]{lemma}{independentSetExtension}\label{lem:independent_set_extension}
For every subgraph $H\subseteq G$ and every edge $e=(u,v)\in E(H)$,
\begin{equation} \label{eq:independent_set_inequality}
    \Pr[\mathcal{I}_H]
    \leq
    (1-p^2)
    \Pr[\mathcal{I}_{H\setminus\{e\}}]^{q}
    \Pr[\mathcal{I}_{H-u-v}]^{p}.
\end{equation}
Equivalently,
\begin{equation} \label{eq:conditional_independent_set_extension}
    \Pr[\mathcal{I}_H\mid\mathcal{I}_{H-u-v}]
    \leq
    (1-p^2)
    \Pr[\mathcal{I}_{H\setminus\{e\}}\mid
        \mathcal{I}_{H-u-v}]^{q}.
\end{equation}
\end{restatable} 
We now are ready to complete the proof of \Cref{lem:constructed_residual_function}.
\begin{proof}[Proof of \Cref{lem:constructed_residual_function}]
Since $\Pr[\mathcal{I}_H]\leq 1$, the function $\scr{S}$ is non-negative.
If $e=(u,v)\in E(H)$, then deleting $e$ removes one of the constraints
defining $\mathcal{I}_H$.  Thus
$\mathcal{I}_H\subseteq\mathcal{I}_{H\setminus\{e\}}$, and hence
$
    \Pr[\mathcal{I}_{H\setminus\{e\}}]
    \geq
    \Pr[\mathcal{I}_H].
$
Since $x\mapsto-\ln x$ is decreasing, it follows from
\eqref{eq:residual_function_construction} that
$
    \scr{S}(H\setminus\{e\})\leq\scr{S}(H).
$
This is \eqref{eq:fo_query_reject}.
Next, by \Cref{lem:independent_set_extension} and the elementary inequality
$1-p^2\leq\mathrm e^{-p^2}$,
\eqref{eq:probability_inequality_sufficient_for_residual} holds.  As observed
above, this is equivalent to \eqref{eq:fo_query_accept_or_reject}.  Finally,
\eqref{eq:residual_function_initial_value} verifies
\eqref{eq:fo_initial_value}.  Thus all of the conditions in
\Cref{def:residual_function} are satisfied.
\end{proof}

\subsection*{Acknowledgements}
This work was done in part while the authors were
visiting the Simons Institute for the Theory of Computing for the programs
on \textit{Data-Driven Decision Processes} and \textit{Graph Limits and Processes on Networks: From Epidemics to Misinformation}.
The authors thank anonymous reviewers from the Twenty-Fifth ACM Conference on \textit{Economics and Computation (EC'24)} for useful comments about positioning.

\subsection*{AI Disclosure}
The surplus function of \Cref{sec:designing_residual_function} was found with the help of ChatGPT 5.6 Pro.

\bibliographystyle{abbrvnat}
\bibliography{refs}


\appendix

\section*{E-Companion}

\section{Deferred Proofs from \Cref{sec:ocrs_positive}} \label{sec:OCRS_pos_additions}

\subsection{Proof of \Cref{lem:tree_reduction_game_part_1}} \label{pf:lem:tree_reduction_game_part_1}
{\color{blue}\treeReductionGamePartOne*}
We roughly follow the argument of \citet{kulkarni2022}. For the tree $T$, define the witness tree $W(T)$. An edge $f$ of $T \cup \partial T$ is in $W(T)$ if it is observed before every edge $f'$ that lies above it in $T$.\footnote{$W(T)$ is not literally a tree because it might have edges between $\revcolor{D^{\ell}(e)}$ and $\revcolor{D^{\ell+1}(e)}$ which create cycles. Here, we follow the notation in \citet{kulkarni2022}.}   By definition, the edges of $W(T)$ are observed in a bottom-up fashion. Note that whether the \Cref{alg:tree_aom} matches $e$ is a function of the edges in $W(T)$, and the decision it makes on those edges. Thus, \revcolor{as far as} bounding $\Pr[e \in \scr{M} \mid X_e =1]$ is concerned, we can assume that $T \cup \partial T = W(T)$. 

A few notes about this identification: identifying $T \cup \partial T$ with $W(T)$ might replace some of the $x_e$ values with $0$. For instance, if a vertex in $W(T)$ is a leaf, we think about it as a leaf of $T \cup \partial T$ after the identification. This might reduce the values of some of the $x_e$, but this is safe, since all we ever assume about the $x_e$ in our arguments is that $\sum_{e \in \partial (u)}x_{e} \leq 1$, and this constraint remains satisfied even after reducing some $x_e$.

Furthermore, if $W(T)$ does not have an edge that goes from $\revcolor{D^{\ell}(e)}$ to $\revcolor{D^{\ell+1}(e)}$, there is nothing to prove, since in this case, the neighborhood of $e$ is effectively a tree, so $\Pr[e \in \scr{M} \mid X_e =1] \geq \alpha$. We will therefore assume that $T$ still has depth $\ell$ even after $T \cup \partial T$'s identification with $W(T)$. This establishes the lemma.

\subsection{Proof of \Cref{lem:error_decrease}} \label{pf:lem:error_decrease}

\errorDecrease*

{\color{blue} We first establish the following preliminary proposition:
\begin{proposition} \label{lem:simple_optimization}
Fix $k \in \mb{N}$, $\alpha \in (0,1)$, and $\eps,\bar{x} \in [0,1]$. Suppose $\bm{x}=(x_i)_{i=1}^k$ has non-negative entries, and $\sum_{i=1}^k x_i = \bar{x}$. Then, by maximizing (respectively, minimizing) over all such $\bm{x}$:
\[
\max_{\bm{x}} \prod_{i=1}^{k} \left(1 + \frac{\eps \alpha x_i}{1 - \alpha \sum_{j \le i}x_j} \right) = 1 + \frac{ \eps \alpha \bar{x}}{1 - \alpha \bar{x}},  
\, \min_{\bm{x}} \prod_{i=1}^{k} \left(1 - \frac{\eps \alpha x_i}{1 - \alpha \sum_{j \le i}x_j} \right) = 1 - \frac{ \eps \alpha \bar{x}}{1 - \alpha \bar{x}}.  
\]
\end{proposition}
}
\begin{proof}
Let us concentrate first on establishing the maximizing version of the \revcolor{proposition}. We prove the statement by induction on $k$. First, let us show that the statement is true for $k = 2$. Note that:
\begin{align*}
  &\qquad& \left(1 + \frac{\eps \alpha x_1}{1 - \alpha x_1} \right)\cdot\left(1 + \frac{\eps \alpha x_2}{1 - \alpha (x_1+x_2)} \right) &\leq 1 + \frac{ \eps \alpha (x_1+x_2)}{1 - \alpha (x_1+x_2)} \\
 \iff&& 1 + \frac{\eps \alpha x_1}{1 - \alpha x_1} + \frac{\eps \alpha x_2}{1 - \alpha (x_1+x_2)} + \frac{\eps^2 \alpha^2 x_1x_2}{(1 - \alpha x_1)(1 - \alpha (x_1+x_2))} &\leq 1 + \frac{ \eps \alpha (x_1+x_2)}{1 - \alpha (x_1+x_2)}  \\
\iff&& \frac{\eps \alpha x_1}{1 - \alpha x_1} + \frac{\eps^2 \alpha^2 x_1x_2}{(1 - \alpha x_1)(1 - \alpha (x_1+x_2))} &\leq \frac{ \eps \alpha x_1}{1 - \alpha (x_1+x_2)} \\
\iff&& \eps \alpha x_1(1 - \alpha (x_1+x_2)) + \eps^2 \alpha^2 x_1x_2 &\leq \eps \alpha x_1 (1 - \alpha x_1) \\
\iff&&  - \eps \alpha^2 x_1x_2 + \eps^2\alpha^2 x_1x_2 &\leq   0 \\
\iff&&   \eps \alpha^2 x_1x_2 (\eps -1) &\leq 0.
\end{align*}
It follows immediately that the inequality we started with is true, and furthermore, that the inequality is in fact an equality when $x_1$ or $x_2$ is equal to 0. Hence the statement is true for $k = 2$, and the base case is established.

Suppose we have established the statement for $k = r$, and we are attempting to establish the statement for $k = r+1$. Then, notice that:

\begin{align*}
  \prod_{i=1}^{r+1} \left(1 + \frac{\eps \alpha x_i}{1 - \alpha \sum_{j \le i}x_j} \right) &\leq \left(1 + \frac{ \eps \alpha (x_1+x_2)}{1 - \alpha (x_1+x_2)}\right)\prod_{i=3}^{r+1} \left(1 + \frac{\eps \alpha x_i}{1 - \alpha \sum_{j \le i}x_j} \right) \\
 &\leq 1 + \frac{ \eps \alpha \bar{x}}{1 - \alpha \bar{x}}.
\end{align*}
The first inequality follows from an application of the case $k = 2$, and the second follows from the inductive hypothesis. Furthermore, the inequality is tight when $x_2 = x_3 = \ldots = x_n = 0$. The desired result follows.

The proof of the minimizing version of the \revcolor{proposition} is almost entirely identical, but we include a proof for completness. We prove the statement by induction on $k$. First, let us show that the statement is true for $k = 2$. Note that:

\begin{align*}
  &\qquad& \left(1 - \frac{\eps \alpha x_1}{1 - \alpha x_1} \right)\cdot\left(1 - \frac{\eps \alpha x_2}{1 - \alpha (x_1+x_2)} \right) &\geq 1 - \frac{ \eps \alpha (x_1+x_2)}{1 - \alpha (x_1+x_2)} \\
 \iff&& 1 - \frac{\eps \alpha x_1}{1 - \alpha x_1} - \frac{\eps \alpha x_2}{1 - \alpha (x_1+x_2)} + \frac{\eps^2 \alpha^2 x_1x_2}{(1 - \alpha x_1)(1 - \alpha (x_1+x_2))} &\geq 1 - \frac{ \eps \alpha (x_1+x_2)}{1 - \alpha (x_1+x_2)}  \\
\iff&& -\frac{\eps \alpha x_1}{1 - \alpha x_1} + \frac{\eps^2 \alpha^2 x_1x_2}{(1 - \alpha x_1)(1 - \alpha (x_1+x_2))} &\geq -\frac{ \eps \alpha x_1}{1 - \alpha (x_1+x_2)} \\
\iff&& -\eps \alpha x_1(1 - \alpha (x_1+x_2)) + \eps^2 \alpha^2 x_1x_2 &\geq -\eps \alpha x_1 (1 - \alpha x_1) \\
\iff&&  \eps \alpha^2 x_1x_2 + \eps^2\alpha^2 x_1x_2 &\geq   0.
\end{align*}
It follows immediately that the inequality we started with is true, and furthermore, that the inequality is in fact an equality when $x_1$ or $x_2$ is equal to 0. Hence the statement is true for $k = 2$, and the base case is established.

Suppose we have established the statement for $k = r$, and we are attempting to establish the statement for $k = r+1$. Then, notice that:

\begin{align*}
  \prod_{i=1}^{r+1} \left(1 - \frac{\eps \alpha x_i}{1 - \alpha \sum_{j \le i}x_j} \right) &\geq \left(1 - \frac{ \eps \alpha (x_1+x_2)}{1 - \alpha (x_1+x_2)}\right)\prod_{i=3}^{r+1} \left(1 - \frac{\eps \alpha x_i}{1 - \alpha \sum_{j \le i}x_j} \right) \\
 &\geq 1 - \frac{ \eps \alpha \bar{x}}{1 - \alpha \bar{x}}.
\end{align*}
The first inequality follows from an application of the case $k = 2$, and the second follows from the inductive hypothesis. Furthermore, the inequality is tight when $x_2 = x_3 = \ldots = x_n = 0$. The desired result follows. \end{proof}

\revcolor{We are ready to return to the proof of \Cref{lem:error_decrease}:}
\begin{proof}[Proof of \Cref{lem:error_decrease}]
Fix $r \in \revcolor{D^{\ell-j -1}(e)}$, and assume that $r$ has $k \ge 1$ children, as otherwise $\eps_r =0$. Then, since
$|\eps_{r_i}| \le \eps^{\max}_{\revcolor{j}}$ for $i \in [k]$,
\begin{align*}
\eps_r &= 1 - \prod_{i=1}^{k}\left( 1  + \frac{\eps_{r_i} \alpha x_{e_i}}{1 - \alpha \sum_{i' \le i} x_{r,r_{i'}}} \right) \\
        &\le 1 - \prod_{i=1}^{k}\left( 1  - \frac{\eps_{\revcolor{j}}^{\max} \alpha x_{e_i}}{1 - \alpha \sum_{i' \le i} x_{r,r_{i'}}} \right) \le 1 - \left( 1 - \frac{ \eps_{\revcolor{j}}^{\max} \alpha x_r}{1 - \alpha x_r} \right) \leq \eps^{\max}_{\revcolor{j}} \left(\revcolor{1-\alpha}\right),
\end{align*}
where the second inequality uses \Cref{lem:simple_optimization}, and the last uses $x_r \le 1$. Similarly,
$\eps_r \ge - \eps_{\revcolor{j}}^{\max} \left(\revcolor{1-\alpha}\right)$, and so
$|\eps_{r}| \le \eps_{\revcolor{j}}^{\max}\cdot \left(\revcolor{1-\alpha}\right)$ for all $r \in \revcolor{D^{\ell-j -1}(e)}$.
\end{proof}

\subsection{Proof of \Cref{prop:two_round_coupling}} \label{pf:prop:two_round_coupling}
{\color{blue} \coupling*}
We first handle the case when $x = yz$. In this case, clearly $y, z \ge x$.

If $X = 1$, then set $Y = Z =1$. 
Otherwise, if $X=0$, there are three outcomes to assign probability
mass to: $Y =1$ and $Z=0$, $Y=0$ and $Z=1$, and $Y=0$ and $Z=0$.
Denote the amount assigned to each by $a_{1,0}, a_{0,1}$ and $a_{0,0}$, respectively. We wish
to determine $a_{i,j} \in [0,1]$ for which the following hold:
\begin{enumerate}
    \item $a_{0,0} + a_{0,1} + a_{1,0} =1$
    \item $(1 - x) a_{1,0} + x = y$
    \item $(1 - x) a_{0,1} +  x = z$
\end{enumerate}
Observe that $a_{1,0} = (y -x)/(1 -x)$ and $a_{0,1} = (z -x)/(1- x)$. Clearly each
is within $[0,1]$, as $y-x, z-x \ge 0$. Moreover, $a_{0,0} + a_{0,1} + a_{1,0} = (y + z - 2x)/(1 - x) + a_{0,0}$.
So it suffices to show that $(y + z - 2x)/(1 - x) \le 1$ to conclude $a_{0,0} \in [0,1]$. But this is equivalent
to $y + z \le 1 + x = 1 + yz$. Since $y + z \le 1 + yz$ for any $y, z \in [0,1]$,
$a_{0,0}= 1 - a_{0,1} - a_{1,0} \in [0,1]$.

Finally, observe that $Y$ and $Z$ are independent. This is because
$\Pr[Y=1, Z=1] = \Pr[X=1] = x = y z$.

Let us now suppose that $x \ge y z$. We first couple $X' \sim \Ber(yz)$
and such that $X \ge X'$. By applying the previous construction
to $X'$, we get $Y$ and $Z$ which satisfy $X \ge X' = Y Z$.

\subsection{Proof of \Cref{lem:conditional}} \label{pf:lem:conditional}
{\color{blue}\conditionalIndependence*}
Observe that the fractional matching $(z'_e)_{e \in E}$ is a function
of $(Y_e)_{e \in E}$ (i.e., it is $(Y_e)_{e \in E}$-measurable). Thus,
if we define the random subset $\scr{S}:= \{e \in E: z'_e > 0\}$,
then $\scr{S}$ is also a function of $(Y_e)_{e \in E}.$
On the other hand, observe that once $\scr{S}$ is fixed, $(Z'_e)_{e \in \scr{S}}$
and $(Z_e)_{e \in \scr{S}}$ are equal, and so distributed identically.
Since the random variables $((Y_e)_{e \in E}, (Z_e)_{e \in E})$ are independent by construction,
and $Z'_{e} =0$ for all $e \in E \setminus \scr{S}$, the claim follows.

\subsection{Proof of \Cref{lem:good_fractional_matching}} \label{pf:lem:good_fractional_matching}
{\color{blue}\goodFractionalMatching*}
Consider the randomly generated fractional matching $(z'_f)_{f \in E}$ output
by \Cref{alg:fractional_matching}. Given $e=(u,v) \in E$,
define $z'_{v}(e) := \sum_{\substack{f \in \partial(v):\\ f < e}} z'_f$
and $z'_{u}(e) := \sum_{\substack{f \in \partial(u):\\ f < e}} z'_f$. 
Observe
that $z'_{e} = 1/\log(1/\eps)$ if and only if $Y_e =1$ and $\max\{z'_{v}(e), z'_{u}(e) \} \le 1 - \frac{1}{\log(1/\eps)}$. On the other hand, since $z'_{v}(e)$ and $z'_{u}(e)$ are independent 
of $Y_e$, to prove the lemma it suffices to show that
\begin{equation} \label{eqn:control_endpoints}
    \Pr\left[\max\{z'_{v}(e), z'_{u}(e) \} \le 1 - \frac{1}{\log(1/\eps)} \right] \ge 1 - 2\exp\left(- \Omega(\log^{1/2}(1/\eps))\right)
\end{equation}
We begin by considering $z'_{v}(e)$.
Observe that 
\begin{equation} \label{eqn:e_arrival}
    z'_{v}(e) = \sum_{\substack{f \in \partial(v):\\ f < e}} z'_f \le \sum_{\substack{f \in \partial(v) \setminus e}} \frac{Y_f}{\log(1/\eps)},
\end{equation}
where $(Y_f)_{f \in \partial(v) \setminus e}$ are independently drawn Bernoulli random-variables. Moreover, $$\sum_{\substack{f \in \partial(v) \setminus e}} \mb{E}[ Y_f ] = \sum_{f \in \partial(v) \setminus e} (1 - \log^{-1/4}(\eps^{-1})) x_f \log(1/\eps) \le (1 - \log^{-1/4}(\eps^{-1})) \log(1/\eps),$$
since $\sum_{f \in \partial(v)} x_f \le 1$. Setting $\delta = \log^{-1/4}(\eps^{-1})$ and $\mu = (1 - \log^{-1/4}(\eps^{-1})) \log(1/\eps)$, we can thus
apply Chernoff bound to get that for $n$ sufficiently large,
\begin{equation*}
   \Pr\left[\sum_{f \in \partial(v) \setminus e} Y_f \ge (1 + \delta) \left(1 - \log^{-1/4}(\eps^{-1})\right) \log(1/\eps) \right] \le \exp\left(-\frac{\delta^2 \mu}{2 + \delta} \right) \le \exp\left(-\frac{\log^{1/2}(\eps^{-1})}{3} \right).
\end{equation*}
On the other hand,  $(1 + \delta) \left(1 - \log^{-1/4}(\eps^{-1})\right) \le 1 -\log^{-1}(1/\eps)$, so
$$
\Pr\left[\sum_{f \in \partial(v) \setminus e} Y_f \ge \log(1/\eps)\left(1 - \frac{1}{\log(1/\eps)}\right) \right] \le \exp\left(- \frac{\log^{1/2}(\eps^{-1})}{3}\right).
$$
Thus, $\Pr[ z_{v}(e) > 1 - 1/\log(1/\eps)] \le \exp\left(- \Omega(\log^{1/2}(1/\eps))\right)$ after applying this bound to
\eqref{eqn:e_arrival}. An analogous argument applies to $z'_{u}(e)$, and so \eqref{eqn:control_endpoints} follows after applying a union bound. The proof of the lemma
is thus complete.

\subsection{Proof of \Cref{lem:cycle_upper_bound}} \label{pf:lem:cycle_upper_bound}
{\color{blue}\cycleUpperBound*}
In order to prove the lemma, we first upper bound the probability there exists
a cycle in $N^{\ell}_{G'}(e)$ which contains $e$.
Afterwards, we upper bound the probability there exists a cycle in $N^{\ell}_{G'}(e)$ which does \textit{not} contain
$e$. Both upper bounds are of the form $O( \eps ( \log(1/\eps))^{2\ell+1})$, so a simple union bound will then complete
the proof.

Let $N_k$ be a random variable counting the number of $k$ length cycles in $G'$ which contain $e$. Since we only care about cycles of this form in the $\ell$-neighborhood of $e$,
it suffices to prove that
\begin{equation} \label{eqn:e_contained_probability}
\Pr\left[\sum_{k = 3}^{2\ell+2}N_k > 0 \Big| Y_e = 1 \right] =O\left(\eps(\log(1/\eps))^{2\ell+1}\right).
\end{equation}
Assume for now that
\begin{equation} \label{eqn:cycles_length_k}
\E[N_k \mid Y_e =1] \leq \eps(\log(1/\eps))^{k-1}. 
\end{equation}
In this case, by Markov's inequality and linearity of expectation,
\[
\Pr\left[\sum_{k = 3}^{2\ell+2}N_k > 0 \Big| Y_e = 1 \right] \le \sum_{k = 3}^{2\ell + 2} \eps(\log(1/\eps))^{k-1} = O\left(\eps(\log(1/\eps))^{2\ell+1}\right),\]
where the final inequality holds from the geometric series formula.


In order to establish \eqref{eqn:cycles_length_k},
let $w_1, w_2, \ldots, w_{k-2}$ denote $k - 2$ distinct vertices in $G$, all of them different from $u$ and $v$. For convenience, let us write $u = w_0$, and $v = w_{k-1}$. Then, we have that:

\[\E[N_k \mid Y_e =1] = \sum_{w_1, w_2, \ldots, w_{k-2}}\prod_{i = 0}^{k-2}y_{w_iw_{i+1}} \leq  (\log(1/\eps))^{k-1} \sum_{w_1, w_2, \ldots, w_{k-2}}\prod_{i = 0}^{k-2}x_{w_iw_{i+1}},\]
since $y_e \leq x_e \log(1/\eps)$. Now, we know $x_{w_{k-2}v} \leq \epsilon$, for all choices of $w_{k-2}$. Therefore, it follows that:

\[\E[N_k \mid Y_e =1] \leq  \eps (\log(1/\eps))^{k-1} \sum_{w_1, w_2, \ldots, w_{k-2}}\prod_{i = 0}^{k-3}x_{w_iw_{i+1}}.\]
Finally, we will show using a simple induction on $k$ that
\begin{equation}\label{eqn:inductive_bound_product}
\sum_{w_1, w_2, \ldots, w_{k-2}}\prod_{i = 0}^{k-3}x_{w_iw_{i+1}} \leq 1.
\end{equation}
The statement is true for $k = 3$, since $\sum_{w_1} x_{uw_1} \leq 1$ as $\bm{x}$ is fractional matching. Suppose we have already established the statement for $k = r \geq 3$. Then, for $k = r+1$,
\begin{align*}
\sum_{w_1, w_2, \ldots, w_{r-1}}\prod_{i = 0}^{r-2}x_{w_iw_{i+1}} &=  \sum_{w_1, w_2, \ldots, w_{r-2}}\prod_{i = 0}^{r-3}x_{w_iw_{i+1}}\cdot \left(\sum_{w_{r-1}}x_{w_{r-2}w_{r-1}}\right) \\
&\leq \sum_{w_1, w_2, \ldots, w_{r-2}}\prod_{i = 0}^{r-3}x_{w_iw_{i+1}}
\leq 1,
\end{align*}
where the first inequality follows from the fact that $\bm{x}$ is a fractional matching, and the second from the induction hypothesis. We conclude that \eqref{eqn:inductive_bound_product} 
as desired, and so \eqref{eqn:cycles_length_k} holds.

Next, let us consider cycles in $N^{\ell}_{G'}(e)$ which do \textit{not} contain the edge $e$. The existence of such a cycle guarantees the existence of a path of length $k \leq \ell - 1$ from either $u$ of $v$ to a distinguished vertex $w$ (possibly equal to $u$ or $v$), and a cycle of length at most $2(\ell-k) + 1$ that contains $w$. Note that we can upper bound the probability that such a structure by:
$$\sum_{w \in V, k \leq \ell-1}\Pr[\text{$\exists$ a path of length $k$ from $u$ to $w$}]\Pr[\text{$\exists$ a cycle of length $\leq 2(\revcolor{\ell}-k) + 1$ containing $w$}].$$
Now, first of all, 
$$\Pr[\text{there is a cycle of length $\leq 2(\revcolor{\ell}-k) + 1$ containing $w$}] = O(\eps (\log(1/\eps))^{2(\revcolor{\ell}-k)+1}),$$
by an argument analogous to the one used to establish \eqref{eqn:e_contained_probability} (we let $w_0 = w_{\revcolor{\ell}-k-1} = w$). Furthermore, note that if we write $u = w_0$, and $w = w_{k}$, then:
\begin{align*}
\sum_{w \in G}\Pr[\text{there is path of length $k$ from } u \text{ or } v \text{ to } w] 
    &\leq2\sum_{w \in G}\Pr[\text{there is path of length $k$ from } u \text{ to } w] \\
    &\leq 2\sum_{w_1, w_2, \ldots, w_{k-1}, w_{k}}\prod_{i = 0}^{k-1}y_{w_iw_{i+1}} \\
    &\leq 2(\log(1/\eps))^{k}\sum_{w_1, w_2, \ldots, w_{k-1}, w_{k}}\prod_{i = 0}^{k-1}x_{w_iw_{i+1}} \\
    &\leq 2(\log(1/\eps))^{k}.    
\end{align*}
where the final inequality follows from \eqref{eqn:inductive_bound_product}. Finally it follows that the probability with which such a structure exists is at most:
$$\sum_{k = 0}^{\ell -1} O( \eps (\log(1/\eps))^{2(\ell-k)+1} )\cdot 2(\log(1/\eps))^{k} = O\left(\eps(\log(1/\eps))^{2 \ell +1}\right).$$
Together with the previous upper bound of \eqref{eqn:e_contained_probability}, this completes the proof.

\section{Deferred Proofs from \Cref{sec:OCRS_neg}} \label{sec:OCRS_neg_additions}
\subsection{Proof of \Cref{lem:small_variance}} \label{pf:lem:small_variance}
{\color{blue}\smallVariance*}
Let suppose that $v_i \neq v_j$ satisfy $\theta = \Cov[A_{v_i},A_{v_j}]$.
Consider when $(v_i,v_j)$  is presented to the OCRS in the second phase,
and $(v_i,v_j)$ is active. Observe that in order for $(v_i,v_j)$ to be matched,
both $v_i$ and $v_j$ must be \textit{available} (i.e., $A_{v_i} = A_{v_j} =1$). Thus, since
$A_{v_i}$ and $A_{v_j}$ are independent of $X_{v_i,v_j}$,
\begin{align*}
    c = \Pr[(v_i,v_j) \in \scr{M} \mid X_{v_i,v_j} =1] \le \Pr[A_{v_i} =1, A_{v_j} =1] &= \mb{E}[A_{v_i}] \mb{E}[A_{v_j}] + \Cov[A_{v_i},A_{v_j}] \\
    & \le \left( 1 - c \left(1 - \frac{2}{n}\right) \right)^2 + \theta,
\end{align*}
where the final inequality uses \eqref{eqn:availability} together with $A_{v_i} = 1 - M_{v_i}$,
as well as the definition of $\theta$. As such, after rearranging for $\theta$ and taking $n \rightarrow \infty$,
\begin{equation*} 
    \theta \ge c - \left( 1 - c \left(1 - \frac{2}{n}\right) \right)^2 = (1 +o(1))(c - (1-c)^2).
\end{equation*}

\subsection{Proof of \Cref{lem:large_variance}} \label{pf:lem:large_variance}
{\color{blue}\largeVariance*}
In order to prove the lemma, we analyze the third and final phase.
Let us assume that the edges of $\partial(u)$ are presented in order $(u,v_1), \ldots , (u,v_n)$ to the OCRS,
and consider the arrival of the last edge $e_n = (u,v_n)$. Define $M_{u}(e_n)$
to be the indicator random variable for the event $u$ is matched when $e_n$ arrives. Observe then
that conditional on $X_{e_n}=1$, $e_n$ is matched
only if $M_{v_n} =0$ \textit{and} $M_{u}(e_n) =0$. Thus, since $X_{e_n}$ is independent of
the latter two events,
\begin{align*}
    c =\Pr[e_n \in \scr{M} \mid X_{e_n}=1] &\le \Pr[M_{v_n} =0, M_{u}(e_n)=0] \\
                                                &=1 - \Pr[M_{v_n}=1] - \Pr[M_{u}(e_n)=1] + \Pr[M_{v_n} =1, M_{u}(e_n)=1] \\
                                                &=1 - \Pr[M_{v_n}=1](1 -\Pr[M_{u}(e_n)=1 \mid M_{v_n}=1]) - \Pr[M_{u}(e_n)=1]                                         
\end{align*}
Now, $\Pr[M_{u}(e_n) =1] = \sum_{i=1}^{n-1} \Pr[(u,e_i) \in \scr{M}] = c \left(1 - \frac{1}{n}\right)$, so
after simplifying and taking $n \rightarrow \infty$, we get that
\begin{equation} \label{eqn:simple_upper_bound}
2c + \Pr[M_{v_n}=1](1 -\Pr[M_{u}(e_n)=1 \mid M_{v_n}=1]) \le 1 +o(1).
\end{equation}
We shall now prove that
\begin{equation} \label{eqn:joint_matching_probability}
    \Pr[M_{u}(e_n)=1 \mid M_{v_n}=1] \le (1 +o(1))(1 - \exp(c + \theta/\Pr[M_{v_n}=1] -1)).
\end{equation}
This will complete the proof. To see this, first observe that by applying \eqref{eqn:joint_matching_probability} to \eqref{eqn:simple_upper_bound},
$$
2c + \Pr[M_{v_n}=1](\exp(c + \theta/\Pr[M_{v_n}=1] -1)) \le 1 + o(1).
$$
Thus, since $\Pr[M_{v_n}=1] = c(1-2/n)$ by \eqref{eqn:availability},
$2c + c(\exp(c + \theta/c -1)) \le 1 + o(1)$ as required.

In order to prove \eqref{eqn:joint_matching_probability}, recall
that $A_{v_n} = 1 - M_{v_n}$. We shall instead condition on $A_{v_n}=0$,
and upper bound $\Pr[M_{u}(e_n)=1 \mid A_{v_n}=0]$. Consider the vertices amongst $v_1, \ldots, v_{n-1}$
which were available after the first phase. Observe that if we condition on $A_{v_n} =0$, 
then in order for $M_{u}(e_n) =1$, there must be an active edge from $u$ to one of these vertices.
Thus,  
\begin{equation} \label{eqn:matching_upper_bound}
\Pr[M_{u}(e_n) =1 \mid  A_{v_n}=0] \le \Pr[ \cup_{j =1}^{n-1} \{ X_{u,v_j} =1\} \cap \{A_{v_j} =1\} \mid A_{v_n}=0].
\end{equation}
On the other hand, the edge states $(X_{u,v_j})_{j=1}^{n}$ are independent from $(A_{v_i})_{i=1}^{n}$, and so
\begin{align*}
    \Pr[ \cup_{j =1}^{n-1} \{ X_{u,v_j} =1\} \cap \{A_{v_j} =1\} \mid (A_{v_i})_{i=1}^{n-1}, A_{v_n} =0] &= 1 - \prod_{j =1}^{n-1} \left(1 - \frac{A_{v_j}}{n}\right) \\
    &=(1 + o(1)) \left(1 - \exp\left(-\frac{\sum_{j=1}^{n-1}A_{v_j}}{n}\right)\right).
\end{align*}
Now, the function $z \rightarrow 1 - \exp(-z)$ is concave for $z \in [0,1]$. Thus, if we take expectations
over the $(A_{v_i})_{i=1}^{n-1}$ random variables (while conditioning on $A_{v_n}=0$) then Jensen's inequality implies
that
\begin{equation*}
    \Pr[ \cup_{j =1}^{n-1} \{ X_{u,v_j} =1\} \cap \{A_{v_j} =1\} \mid A_{v_n} =0] \le (1 +o(1)) \left(1 - \exp\left(-\frac{\sum_{j=1}^{n-1}\mb{E}[A_{v_j} \mid A_{v_n} =0]}{n} \right) \right),
\end{equation*}
and so combined with \eqref{eqn:matching_upper_bound}, we get that
\begin{equation} \label{eqn:expectation_matching_upper_bound}
    \Pr[M_{u}(e_n) =1 \mid  A_{v_n}=0] \le (1 +o(1))\left(1 - \exp\left(-\frac{\sum_{i=1}^{n-1}\mb{E}[A_{v_i} \mid A_{v_n} =0]}{n} \right) \right).
\end{equation}
As such, we focus on upper bounding $\sum_{i=1}^{n-1}\mb{E}[A_{v_i} \mid A_{v_n} =0]$. We argue
that for each $1 \le i \le n-1$, the events $A_{v_n} = 0$ and $A_{v_i}=1$ are \textit{negatively} correlated,
which we quantify in terms of $\theta$. Specifically, by the definition of $\theta$, $\mb{E}[A_{v_i} \mid A_{v_n} = 1]\Pr[A_{v_n} =1] = \mb{E}[A_{v_i}] \mb{E}[A_{v_n}] + \Cov[A_{v_i},A_{v_n}] \ge \mb{E}[A_{v_i}] \mb{E}[A_{v_n}] + \theta$. Thus,
\begin{align*}
\mb{E}[A_{v_i} \mid A_{v_n} = 0] \Pr[A_{v_n} =0] &= \mb{E}[A_{v_i}] -\mb{E}[A_{v_i} \mid A_{v_n} = 1] \Pr[A_{v_n} =1] \\
                                        &\le \mb{E}[A_{v_i}] - \mb{E}[A_{v_i}] \mb{E}[A_{v_n}] - \theta    \\
                                        &= \mb{E}[A_{v_i}]\Pr[A_{v_n} = 0] - \theta.
\end{align*}
Since $\Pr[A_{v_n} =0] > 0$ for $n > 2$, and $\mb{E}[A_{v_i}] = 1 - c \left(1 - \frac{2}{n}\right)$, it follows that 
\begin{equation} \label{eqn:average_conditional_availability}
    \sum_{i=1}^{n-1}\mb{E}[A_{v_i} \mid A_{v_n} =0] \le (1+o(1)) n\left( (1-c)- \frac{\theta}{\Pr[A_{v_n} =0]}\right).
\end{equation}
By applying \eqref{eqn:average_conditional_availability} to \eqref{eqn:expectation_matching_upper_bound},
$ \Pr[M_{u}(e_n) =1 \mid  A_{v_n}=0] \le (1 +o(1))(1 - \exp(-1+c +\theta/\Pr[A_{v_n} =0] )),$
and so \eqref{eqn:joint_matching_probability} is proven, thereby completing the proof.

\subsection{Proof of \Cref{existence_of_end}} \label{pf:existence_of_end}
{\color{blue}\existenceOfEnd*}
Since
\begin{align*}
    \text{Var} (A) = \sum_{i \neq j}\text{Cov}(A_i, A_j) + \sum_{i = 1}^n \Var(A_i),
\end{align*}
it follows that $\sum_{i \neq j}\text{Cov}(A_i, A_j) \leq \theta_nn^2$. On the other hand, 
\begin{align*}
    \sum_{|S| = m}\text{Var} (A_S) &\leq \sum_{|S| = m}\left(\sum_{i \neq j, i, j \in S}\text{Cov}(A_i, A_j) +\delta m\right) \\
    &= \left(\sum_{|S| = m}\sum_{i \neq j, i, j \in S}\text{Cov}(A_i, A_j)\right)+  \binom{n}{m} \cdot \delta m\\
    &= \sum_{i \neq j}\binom{n-2}{m-2}\text{Cov}(A_i, A_j)+  \binom{n}{m} \cdot \delta m.
\end{align*}
It follows that 
\begin{align*}
    \frac{\sum_{|S| = m}\text{Var} (A_S)}{\binom{n}{m}} &\leq \frac{\binom{n-2}{m-2}}{\binom{n}{m}}\left(\sum_{i \neq j}\text{Cov}(A_i, A_j) \right) +\delta m \\
    &\leq \frac{m(m-1)}{n(n-1)}\theta_nn^2 +\delta m\\
    &\leq 2\theta_nm^2 + \delta m.
\end{align*}
Next, notice that, 
\begin{align*}
\delta m \leq \theta_nm^2 &\iff  \delta \leq \theta_nm \\
&\iff  \delta \leq \theta_n^{5/4}n \\ 
&\iff  \delta \leq \theta_n^{5/4}n \\
&\iff \frac{\delta^{4/5}}{n^{4/5}} \leq \theta_n.
\end{align*}
Since we assumed that $\theta_n \geq \frac{1}{\sqrt{n}}$, the last inequality is true for sufficiently large $n$, so we conclude that:
\begin{align*}
    \frac{\sum_{|S| = m}\text{Var} (A_S)}{\binom{n}{m}} \leq 3\theta_nm^2.
\end{align*}
Since the average value of $\text{Var}(A_S)$ is $\leq 3 \theta_nm^2$, it follows that there exists $S$ such that $\text{Var}(A_S) \leq 3\theta_nm^2$, as required.

\subsection{Proof of \Cref{lem:concentrated_ocrs_analysis}} \label{pf:lem:concentrated_ocrs_analysis}
\concentratedOCRSAnalysis*
{\color{blue}
First, note that 
\[\Pr[\mathcal{B}\cap\mathcal{C}]
=\E\left[\bI_{\mathcal{C}}\left(1-\left(1-\frac{1}{n}\right)^{A_{end}}\right)\right],\]
since $\mathcal{B}$ occurs exactly when $X_{u,v}=1$ for some $v$ with $A_{u,v} =1$, and conditional on $A_{end}$ and $\mathcal{C}$, the activity statuses of the last $m$ edges are still independent, with probability $1/n$ of being active. Fixing a threshold $T$, with $\Pr[A_{end} > T] \leq \eps,$ it follows that:
\begin{align*}
\E\left[\bI_{\mathcal{C}}\left(1-\left(1-\frac{1}{n}\right)^{A_{end}}\right)\right]
&\leq \frac{T}{n}\Pr[\mathcal{C}]+\Pr[A_{end}>T] \\
&\leq \frac{T}{n}\Pr[\mathcal{C}]+\eps.
\end{align*}
The second inequality in the lemma follows from $\Pr[\mathcal{C}]\leq 1-c(n-m)/n$.}

\subsection{Proof of \Cref{thm:negative_OCRS_trees}} \label{pf:thm:negative_OCRS_trees}
{\color{blue} \treeUnconditional*}

The proof in this setting will, broadly speaking, resemble the proof of \Cref{lem:large_variance}, with some key differences. We will work with the same sequence of input graphs useful in the proof of \Cref{thm:negative_OCRS_trees_concentrated}, and in the first phase, the edges are presented in the same order. However, we will choose a different order for the edges in $\partial(u)$.

Once again, let us fix an OCRS for $(G,\bm{x})$ which
returns a matching $\scr{M}$, and is $c$-selectable for $c \in (0,1)$, i.e., $\Pr[e \in \scr{M} \mid X_e =1] = c$ for each edge $e$ of $G$. We will assume that $\Var[\sum_{i = 1}^n A_{v_i}] \geq \theta n^2$, for some $\theta \geq 0$. Before we explain the order in which the edges of $\partial(u)$ are presented, we note that since

\[\sum_{j = 1}^n\left(\E [A_{v_j}\sum_{i = 1}^n A_{v_i}] -\E [A_{v_j}] \E[\sum_{i = 1}^n A_{v_i}]\right) = \Var[\sum_{i = 1}^n A_{v_i}] \geq \theta n^2, \]
we can find $j$ such that 
\[\E [A_{v_j}\sum_{i = 1}^n A_{v_i}] -\E [A_{v_j}] \E[\sum_{i = 1}^n A_{v_i}]\geq \theta n. \]
But note that
\begin{eqnarray*}
 && \E [A_{v_j}\sum_{i = 1}^n A_{v_i}] -\E [A_{v_j}] \E[\sum_{i = 1}^n A_{v_i}]\geq \theta n \\
 &\iff& \E [\sum_{i = 1}^n A_{v_i} \mid A_{v_j} = 1]\Pr[A_{v_j} = 1] -(1-c)^2n\geq \theta n \\
  &\iff& \E [\sum_{i = 1}^n A_{v_i}] - \E [\sum_{i = 1}^n A_{v_i}  A_{v_j} = 0]\Pr[A_{v_j} = 0] -(1-c)^2n\geq \theta n \\
    &\iff& (1-c)n - \E [\sum_{i = 1}^n A_{v_i} \mid A_{v_j} = 0]c -(1-c)^2n\geq \theta n \\
        &\iff& \E [\sum_{i = 1}^n A_{v_i} \mid A_{v_j} = 0]c\leq c(1-c)n - \theta n \\
            &\iff& (1-c)n - \E [\sum_{i = 1}^n A_{v_i} \mid A_{v_j} = 0]c -(1-c)^2n\geq \theta n \\
        &\iff& \sum_{i \neq j}\E [ A_{v_i} \mid A_{v_j} = 0]\leq n\left((1-c) - \frac{\theta}{c}\right). \\
\end{eqnarray*}
We conclude that we can find $j$ such that $\sum_{i = 1}^{n-1}\E [ A_{v_i} \mid A_{v_j} = 0]\leq n\left((1-c) - \frac{\theta}{c}\right).$ Let us assume that in phase 2, the edges of $\partial(u)$ are presented in an order that ensures that $v_j$ comes last.

For convenience, let us immediately relabel the vertices so that the edges of $\partial(u)$ are presented in order $e_1 = (u,v_1), \ldots , e_n =  (u,v_n)$ to the OCRS. Notice that by choosing the order of the edges in $\partial(u)$ correctly, we have already established equation \eqref{eqn:average_conditional_availability} from the proof of Lemma \ref{lem:large_variance}. But this equation was the only requirement in the proof of Lemma \ref{lem:large_variance} to establish its conclusion that:
$$2c + c \exp(c +\theta/c -1) \le 1 + o(1).$$
Note that the $\theta$ here has a different meaning than in the proof of \Cref{lem:large_variance}---here, $\theta$ is just a number that satisfies the inequality $\Var[\sum_{i = 1}^n A_{v_i}] \geq \theta n^2$. Clearly this inequality is always true for $\theta = 0$, so we conclude that:
\[2c + c\exp(c-1) \leq 1 + o(1),\]
as was needed.

\section{Proof of \Cref{thm:reductiontogw}} \label{sec:RCRS_additions}
\reductiontogw*

In order to prove \Cref{thm:reductiontogw}, we first restate the following result proven
by \citet{nuti2023towards} (this is Lemma $1$ in \citep{nuti2023towards}). 

\begin{lemma}[\citealp{nuti2023towards}]\label{lemma:tree}
Let $T_0$  be any finite realization of $\T$. For any edge $e$ of $G$:
$$| \Pr[S_e =  T_0 \mid X_e = 1] - \Pr[\T=T_0]| \leq  3\varepsilon(n) |T_0|^2.$$
\end{lemma}
Note that $S_e = T_0$ means that there is an isomorphism from the graphs $S_e$ to $T_0$ which, if $e = (u_0, v_0)$, sends $u_0$ to the special vertex $u \in T_0$ and $v_0$ to the special vertex $v \in T_0$.

We next establish the key proposition needed to prove \Cref{thm:reductiontogw}:
\begin{proposition}\label{prop:lemmacorrs} 
Fix any set of finite trees $F = \{T_1, T_2, \ldots, T_m\}$, the following hold:
\begin{enumerate}
    \item $|\Pr[S_e \notin F\mid X_e = 1]-\Pr[\T \notin F]|\leq 3\varepsilon \sum|T_i|^2.$
    \item $|\Pr[e \in \scr{M}, S_e \in F\mid X_e = 1] - \Pr[(u, v) \in \scr{M}_{\T}, \T\in F]| \leq 3\varepsilon \sum|T_i|^2.$
\end{enumerate}
\end{proposition}
\begin{proof}[Proof of \Cref{prop:lemmacorrs}]
For the first result, note that \Cref{lemma:tree} implies
\begin{equation*}
|\Pr[S_e \in F\mid X_e = 1]-\Pr[\T \in F]| \leq 3\varepsilon \sum|T_i|^2
\end{equation*}
and hence it follows that
\begin{equation*}
|\Pr[S_e \notin F\mid X_e = 1]-\Pr[\T \notin F]|\leq 3\varepsilon \sum|T_i|^2.
\end{equation*}
For the second result, notice that since \textsc{Greedy-CRS} only depends on
the active edges of $G$,
\begin{equation*}
\Pr[e \in \scr{M}\mid X_e = 1, S_e = T_i] = \Pr[(u, v) \in \scr{M}_{\T}\mid \T = T_i].   
\end{equation*}
Hence, it follows that
\begin{eqnarray*}
     &&|\Pr[e \in \scr{M}, S_e \in F\mid X_e = 1] - \Pr[(u, v) \in \scr{M}_{\T}, \T\in F]|\\
    & =& \left|\sum_{i=1}^m\Pr[e \in \scr{M}\mid S_e = T_i, X_e = 1]\Pr[S_e  = T_i \mid X_e = 1] - \sum_{i=1}^m\Pr[(u, v) \in \scr{M}_{\T}|\T= T_i]\Pr[\T = T_i]\right|\\
    & =& \left|\sum_{i=1}^m \Pr[(u, v) \in \scr{M}_{\T}|\T= T_i](\Pr[S_e  = T_i \mid X_e = 1] - \Pr[\T = T_i])\right|\\
    & \leq & \sum_{i=1}^m |\Pr[S_e  = T_i \mid X_e = 1] - \Pr[\T = T_i]|\\
    & \leq& 3\varepsilon \sum|T_i|^2.\\
\end{eqnarray*}
\end{proof}
\begin{proof}[Proof of \Cref{thm:reductiontogw}] \label{pf:thm:reductiontogw}
First of all, note that the triangle inequality implies
\begin{eqnarray*}
     &&|\Pr[e \in \scr{M}\mid X_e = 1] - \Pr[(u, v) \in \scr{M}_{\T}]|\\
    & \leq& |\Pr[e \in \scr{M}, S_e \notin F \mid X_e = 1] - \Pr[(u, v) \in \scr{M}_{\T}, \T \notin F]| \\ &&+ |\Pr[e \in \scr{M}, S_e \in F \mid X_e = 1] - \Pr[(u, v) \in \scr{M}_{\T}, \T \in F]|\\
    & \leq& \Pr[S_e \notin F \mid X_e = 1] + \Pr[\T \notin F] \\&& + |\Pr[e \in \scr{M}, S_e \in F \mid X_e = 1] - \Pr[(u, v) \in \scr{M}_{\T}, \T \in F]|,\\
\end{eqnarray*}
It follows from \Cref{prop:lemmacorrs} that: 
\begin{equation*}
|\Pr[e \in \scr{M}\mid X_e = 1] - \Pr[(u, v) \in \scr{M}_{\T}]| \leq 2\Pr[\T \notin F] + 6\varepsilon \sum|T_i|^2,
\end{equation*}
and the inequality is true for every $F$. Now fix a $\delta > 0$, and consider an $F$ with the property that $\Pr[\T \notin F] < \delta/3.$ Such an $F$ exists since $\T$ is finite with probability 1. Then, find an $\varepsilon$ so small that $6\varepsilon \sum|T_i|^2 < \delta/3$. For such an $\varepsilon$, we must have 
\begin{equation*}
|\Pr[e \in \scr{M}\mid X_e = 1] - \Pr[(u, v) \in \scr{M}_{\T}]| \leq 2\Pr[\T \notin F] + 6\varepsilon \sum|T_i|^2 < \delta,
\end{equation*}
and so Theorem~\ref{thm:reductiontogw} immediately follows.
\end{proof}

\section{{\color{blue} Deferred Proofs from \Cref{sec:free_order_negative}}}

\subsection{Proof of \Cref{lem:assuming_residual}} \label{pf:lem:assuming_residual}

\assumingResidual*

Fix $0\leq t<n^2$ and condition on $\mathcal{F}_t$. Since the FO-CRS
is deterministic, the $(t+1)^{th}$ edge
$
    e_{t+1}=(u_{t+1},v_{t+1})
$
is $\mathcal{F}_t$-measurable. Moreover, since $e_{t+1}$ has not yet been
processed, the independence of the edge states implies that
\begin{equation} \label{eqn:conditional_edge_indep}
    \Pr[X_{e_{t+1}}=1\mid\mathcal{F}_t]=p,
    \qquad
    \Pr[X_{e_{t+1}}=0\mid\mathcal{F}_t]=1-p.
\end{equation}
Suppose first that $e_{t+1}\notin E(R_t)$. By the definition of $R_t$,
the edge $e_{t+1}$ is infeasible and therefore cannot be added to the
matching. Since $e_{t+1} \notin E(R_t)$,
\[
    \mathcal{M}_{t+1}=\mathcal{M}_t,
    \qquad
    R_{t+1}=R_t-e_{t+1} =R_t.
\]
Consequently, $Q_{t+1}=Q_t$. It remains to consider the case $e_{t+1}\in E(R_t)$. Conditional on
$\mathcal{F}_t$, the deterministic FO-CRS either rejects $e_{t+1}$ when it
is active or accepts it when it is active. Suppose first that the FO-CRS rejects $e_{t+1}$ when it is active. Then,
regardless of the value of $X_{e_{t+1}}$,
\[
    \mathcal{M}_{t+1}=\mathcal{M}_t,
    \qquad
    R_{t+1}=R_t \setminus e_{t+1}.
\]
Therefore, by \eqref{eq:fo_query_reject},
\[
    Q_{t+1}
    =
    |\mathcal{M}_t|+\scr{S}(R_t \setminus e_{t+1})
    \leq
    |\mathcal{M}_t|+\scr{S}(R_t)
    =
    Q_t.
\]
Otherwise, the FO-CRS accepts $e_{t+1}$ when it is active. If $X_{e_{t+1}}=0$,
then
\[
    \mathcal{M}_{t+1}=\mathcal{M}_t,
    \qquad
    R_{t+1}=R_t \setminus e_{t+1}.
\]
If $X_{e_{t+1}}=1$, then
\[
    \mathcal{M}_{t+1}=\mathcal{M}_t\cup\{e_{t+1}\},
    \qquad
    R_{t+1}=R_t-u_{t+1}-v_{t+1}.
\]
Thus, after averaging over $X_{e_{t+1}}$ and applying \eqref{eqn:conditional_edge_indep}, we can
write $\mathbb{E}[Q_{t+1}\mid\mathcal{F}_t]$ as
\begin{align*} 
    &\Pr[X_{e_{t+1}}=0\mid\mathcal{F}_t]( |\scr{M}_t| + \scr{S}(R_t \setminus e_{t+1})) + \Pr[X_{e_{t+1}}=1\mid\mathcal{F}_t]( |\scr{M}_t| + 1 + \scr{S}(R_t - u_{t+1} - v_{t+1})) \\
    &=
    |\mathcal{M}_t|
    +(1-p)\scr{S}(R_t \setminus e_{t+1})
    +p\bigl(1+\scr{S}(R_t-u_{t+1}-v_{t+1})\bigr) \\
    &\leq
    |\mathcal{M}_t|+\scr{S}(R_t)
    =
    Q_t,
\end{align*}
where the inequality uses \eqref{eq:fo_query_accept_or_reject}.
Thus, no matter the decision of the FO-CRS when processing $e_{t+1}$,
$
    \mathbb{E}[Q_{t+1}\mid\mathcal{F}_t]\leq Q_t.
$
Hence $(Q_t)_{t=0}^{n^2}$ is a supermartingale with respect to
$(\mathcal{F}_t)_{t=0}^{n^2}$. 

To complete the proof, recall that $Q_t = |\scr{M}_t| + \scr{S}(R_t)$ for each $t =0, \ldots , n^2$. Thus, since $(Q_t)_{t=0}^{n^2}$ is a supermartingale, we can take expecations to get that
\begin{equation} \label{eqn:supermartingale_implication}
   \mathbb{E}[|\mathcal{M}_{n^2}| + \scr{S}(R_{n^2})] = \mathbb{E}[Q_{n^2}] \le \mathbb{E}[Q_0] = \mathbb{E}[|\mathcal{M}_{0}| + \scr{S}(R_{0})].
\end{equation}
On the other hand, $\scr{S}(R_{n^2}) \ge 0$, $|\scr{M}_0| = 0$, 
and $\mathbb{E}[\scr{S}(R_{0})] \le (1+o(1))
    \left(
        1-\ln\left(2-\frac{1}{\mathrm e}\right)
    \right)n$ by \eqref{eq:fo_initial_value}. Combining these inequalities
    with \eqref{eqn:supermartingale_implication} yields
$$
    \mathbb{E}[|\mathcal{M}_{n^2}|]
    \leq
    (1+o(1))
    \left(
        1-\ln\left(2-\frac{1}{\mathrm e}\right)
    \right)n,
$$
as required.

\subsection{Proof of \Cref{lem:independent_set_extension}} \label{pf:lem:independent_set_extension}

\independentSetExtension*

Write
$
    \bm{B}_{-u-v}
    :=
    (B_w)_{w\in V(G)\setminus\{u,v\}},
$
and call $u$ \emph{available} if
$B_w=0$ for every $w\in N_H(u)\setminus\{v\}$, and define the availability of
$v$ analogously. Denote these events by $A_u$ and $A_v$, respectively,
and notice that $A_u$ and $A_v$ are measurable with respect to $\bm{B}_{-u-v}$,
as is the event $\scr{I}_{H-u-v}$. In what follows, we condition on $\bm{B}_{-u-v}$,
and assume that we have an instantiation of it for which $\scr{I}_{H-u-v}$ occurs.
Our goal is to first show that
\begin{equation} \label{eq:pointwise_extension}
    \Pr[\mathcal{I}_H\mid
        \bm{B}_{-u-v}] \cdot \bm{1}_{\scr{I}_{H - u -v}}
    \leq
    (1-p^2)
    \Pr[\mathcal{I}_{H\setminus\{e\}}\mid
        \bm{B}_{-u-v}]^q \cdot \bm{1}_{\scr{I}_{H - u -v}}.
\end{equation}
We separate \eqref{eq:pointwise_extension} into cases, dependent on which combination of $A_u$ and $A_v$ occur.
Note that conditional on $\bm{B}_{-u-v}$, $B_u$ and $B_v$ remain independent $\Ber(p)$ random variables.

Suppose first that $A_u, A_v$ and $\scr{I}_{H-u-v}$ all occur. In this case, for
$\scr{I}_{H\setminus\{e\}}$ to occur, both $B_u$ and $B_v$ can take any value. On the other hand,
for $\scr{I}_H$ to occur, the outcome $B_u=B_v=1$ is forbidden.  Consequently,
\[
    \Pr[\mathcal{I}_H\mid
        \bm{B}_{-u-v}] \cdot \bm{1}_{A_u \cap A_v} \bm{1}_{\scr{I}_{H - u -v}}
    =
    1-p^2,
    \qquad
    \Pr[\mathcal{I}_{H\setminus\{e\}}\mid
        \bm{B}_{-u-v}] \cdot \bm{1}_{A_u \cap A_v} \bm{1}_{\scr{I}_{H - u -v}}
    =
    1,
\]
and so \eqref{eq:pointwise_extension} follows immediately.

Now suppose that $\scr{I}_{H-u-v}$ occurs, and at least one of $u$ and $v$ is unavailable.  In order
for $\scr{I}_{H\setminus\{e\}}$ to occur, if $\neg A_u \cap \neg A_v$ holds, then we must have $B_u = B_v =0$. Similarly,
if $\neg A_u \cap A_v$ (respectively, $A_u \cap \neg A_v$) holds, then we must have $B_u =0$ (respectively, $B_v =0$). 
The exact same cases hold for $\scr{I}_{H}$. Thus,
\begin{equation} \label{eq:equal_extension_probabilities}
    \Pr[\mathcal{I}_H\mid
        \bm{B}_{-u-v}] \cdot \bm{1}_{\neg A_u \cup \neg A_v} \bm{1}_{\scr{I}_{H - u -v}}
    =
    \Pr[\mathcal{I}_{H\setminus\{e\}}\mid
        \bm{B}_{-u-v}] \cdot \bm{1}_{\neg A_u \cup \neg A_v} \bm{1}_{\scr{I}_{H - u -v}}
    \in\{q^2,q\}.
\end{equation}
Weighted AM--GM, with weights $p$ and $q$, gives
$
    q^p=q^p1^q\leq pq+q\cdot 1=1-p^2,
$
and so for every $x\in[0,q]$,
\begin{equation}\label{eq:equal_case}
    x=x^qx^p
    \leq q^px^q
    \leq(1-p^2)x^q.
\end{equation}
By taking $x \in \{q, q^2\}$ we can apply \eqref{eq:equal_extension_probabilities} with \eqref{eq:equal_case},
to get \eqref{eq:pointwise_extension}.

To complete the proof,  we take conditional expectations with respect to $\scr{I}_{H-u-v}$ over \eqref{eq:pointwise_extension} to get that
\begin{align*}
     \Pr[\mathcal{I}_H\mid\mathcal{I}_{H-u-v}] &= \mb{E}\left[\Pr[\mathcal{I}_H\mid
        \bm{B}_{-u-v}] \mid \scr{I}_{H-u-v} \right] \\
    &\leq
    (1-p^2)
    \mb{E}\left[
        \Pr[\mathcal{I}_{H\setminus\{e\}}\mid\bm{B}_{-u-v}]^q
        \,\middle|\,\mathcal{I}_{H-u-v}
    \right]\\
    &\leq
    (1-p^2)
    \mb{E}\left[
        \Pr[\mathcal{I}_{H\setminus\{e\}}\mid\bm{B}_{-u-v}]
        \,\middle|\,\mathcal{I}_{H-u-v}
    \right]^q\\
    &=
    (1-p^2)
    \Pr[\mathcal{I}_{H\setminus\{e\}}\mid
        \mathcal{I}_{H-u-v}]^q.
\end{align*}
Here the first equality uses that $\scr{I}_{H-u-v}$ is measurable with respect to $\bm{B}_{-u-v}$.
The second inequality is \eqref{eq:pointwise_extension} with the conditional expectation applied.
The third inequality uses Jensen's inequality, as $x\mapsto x^q$ is concave.
The final equality once again uses that
$\mathcal{I}_{H-u-v}$ is measurable with respect to $\bm{B}_{-u-v}$.  This proves
\eqref{eq:conditional_independent_set_extension}.  

Since
$\mathcal{I}_H$ and $\mathcal{I}_{H\setminus\{e\}}$ both imply
$\mathcal{I}_{H-u-v}$, we have
\[
    \Pr[\mathcal{I}_H\mid\mathcal{I}_{H-u-v}]
    =
    \frac{\Pr[\scr{I}_H]}{\Pr[\scr{I}_{H -u - v]}},
    \qquad
    \Pr[\mathcal{I}_{H\setminus\{e\}}\mid\mathcal{I}_{H-u-v}]
    =
    \frac{\Pr[\scr{I}_{H\setminus\{e\}}]}{\Pr[\scr{I}_{H -u - v]}}.
\]
Substituting these identities into
\eqref{eq:conditional_independent_set_extension}, multiplying by
$\Pr[\scr{I}_{H -u - v]}$, and using that $1-q=p$ proves
\eqref{eq:independent_set_inequality}.

\section{\revcolor{Concentration Property of OCRSs}}
\label{apx:concentration-property}
\begin{proposition}
\label{prop:independent-coin-concentration}
Fix an input $(G,\bm{x})$ and an arrival order of the edges $E$ of $G$. Suppose that, for
each arriving edge $e$, an OCRS independently draws
\[
    A_e \sim \operatorname{Ber}(a_e),
\]
where $a_e\in[0,1]$ is deterministic once $(G,\bm{x})$ and the arrival
order are fixed. The OCRS adds $e$ to its current matching if and only if
$X_e=A_e=1$ and both endpoints of $e$ were previously unmatched.

Consider a sequence of inputs $(G(n),\bm{x}(n))$ such that if $G(n) = (V(n),E(n))$,
then as $n \rightarrow \infty$,
\[
    s_n:=\sum_{e\in E(n)}x_e(n)\longrightarrow\infty.
\]
If the resulting OCRS is $c_n$-selectable on $(G(n),\bm{x}(n))$, where
$\liminf_{n\to\infty}c_n>0$, then it is concentrated on this sequence
of inputs. That is, if $M_n$ denotes the matching it outputs on $G_n$, then w.h.p.,
\[
    \bigl||M_n|-\mathbb{E}[|M_n|]\bigr|
        =o\bigl(\mathbb{E}[|M_n|]\bigr).
\]
\end{proposition}

\begin{proof}
Fix $n \in \mb{N}$. For convenience, we suppress the dependence on $n$ from the notation. Since
the arrival order is fixed, we may equivalently draw all the random
variables $(A_e)_{e\in E}$ before the execution begins. Define
\[
    Y_e:=A_eX_e.
\]
The random variables $(Y_e)_{e\in E}$ are mutually independent, with
\[
    q_e:=\Pr[Y_e=1]=a_ex_e\leq x_e.
\]
The matching $M$ is precisely the matching obtained by greedily
processing, in the given arrival order, those edges $e$ for which
$Y_e=1$.

We first observe that changing a single coordinate of
$\bm{Y}=(Y_e)_{e\in E}$ changes the cardinality of the resulting
matching by at most one. To see this, consider two executions whose
candidate vectors differ only on one edge $e$. The executions are
identical before $e$ is processed. If the differing value of $Y_e$
does not cause $e$ to be added in exactly one execution, the two
executions remain identical. Otherwise, immediately after processing
$e$, the symmetric difference of the two matchings consists only of
$e$ and is therefore an alternating path.

This property is preserved as the remaining edges are processed.
Indeed, suppose that a later candidate edge $f$ is added in exactly
one of the two executions. Then $f$ is feasible in one execution and
infeasible in the other. Consequently, one of its endpoints is
unmatched in the first matching and matched in the second matching,
and hence is an endpoint of the current alternating path. Adding $f$
therefore extends the alternating path or closes it into an
alternating cycle. Thus, throughout the coupled executions, the
symmetric difference is empty, an alternating path, or an alternating
cycle. The cardinalities of two matchings whose symmetric difference
has this form differ by at most one. Hence, writing $M(\bm{y})$ for
the greedy matching obtained from a deterministic candidate vector
$\bm{y}$, we have
\begin{equation}
\label{eq:greedy-cardinality-lipschitz}
    \bigl||M(\bm{y})|-|M(\bm{y}')|\bigr|\leq 1
\end{equation}
whenever $\bm{y}$ and $\bm{y}'$ differ in a single coordinate.

For each $e\in E$, let $\bm{Y}^{(e)}$ be obtained from $\bm{Y}$ by
replacing $Y_e$ with an independent copy $Y'_e$. Applying the
Efron--Stein inequality and then
\eqref{eq:greedy-cardinality-lipschitz}, we obtain
\begin{align*}
    \operatorname{Var}(|M|)
    &\leq
      \frac{1}{2}\sum_{e\in E}
      \mathbb{E}\left[
        \left(|M(\bm{Y})|-|M(\bm{Y}^{(e)})|\right)^2
      \right] \\
    &\leq
      \frac{1}{2}\sum_{e\in E}\Pr[Y_e\neq Y'_e] \\
    &=\sum_{e\in E}q_e(1-q_e) \\
    &\leq\sum_{e\in E}x_e
      =s_n.
\end{align*}
On the other hand, $c_n$-selectability implies
\[
    \mathbb{E}[|M|]
      =\sum_{e\in E}\Pr[e\in M]
      =\sum_{e\in E}x_e\Pr[e\in M\mid X_e=1]
      \geq c_ns_n.
\]
Let $\delta_n:=s_n^{-1/4}$. Chebyshev's inequality now gives
\begin{align*}
    \Pr\left[
      \bigl||M|-\mathbb{E}[|M|]\bigr|
        >\delta_n\mathbb{E}[|M|]
    \right]
    &\leq
      \frac{\operatorname{Var}(|M|)}
           {\delta_n^2\mathbb{E}[|M|]^2} \\
    &\leq
      \frac{s_n}{\delta_n^2c_n^2s_n^2}
      =\frac{1}{c_n^2\sqrt{s_n}}
      =o(1).
\end{align*}
Since $\delta_n\to 0$, this proves w.h.p.
$
    \bigl||M|-\mathbb{E}[|M|]\bigr|
       =o\bigl(\mathbb{E}[|M|]\bigr).
$
\end{proof}
For the tree input $(G, \bm{x})$ from \Cref{sec:OCRS_neg_trees} with edges $E$, we have
\[
    \sum_{e\in E}x_e
      =n(n-1)\frac{1}{n}+n\frac{1}{n}
      =n.
\]
Consequently, every OCRS which satisfies the properties of
\Cref{prop:independent-coin-concentration} is concentrated on this sequence of inputs.
In particular, \Cref{alg:tree_aom} belongs to this class, since on $(G,\bm{x})$, it independently
draws
\[
    A_e\sim
    \operatorname{Ber}\left(
      \frac{\alpha}
      {(1-\alpha x_u(e))(1-\alpha x_v(e))}
    \right)
\]
and adds $e$ precisely when $X_e=A_e=1$ and its endpoints are
unmatched. This is also true of the algorithm of \cite{Ezra_2020, macrury2023random}, despite the parameters of its relevant $A_e$ bits not having a simple closed form.

We note that the Gibbs distribution based OCRS for matchings introduced by \citet{aminian2026stationary} also satisfies the concentration property, using a different argument. Indeed, the distribution $\mu$ of matchings selected by their algorithm is of the form:
\[
\mu(M)\propto \prod_{e\in M} w_e,
\] 
where $(w_e)_{e \in E}$ are positive real values, and $\propto$ indicates equality up to a normalization constant.
By the Heilmann--Lieb Theorem~ (see \cite{heilmann1972theory}), the cardinality generating polynomial
\[
\sum_{M} z^{|M|}\prod_{e\in M} w_e
\]
has only real negative roots. It follows that $|M|$ is distributed as a sum of independent Bernoulli random variables, and hence
\[
\operatorname{Var}(|M|)\leq \mathbb{E}|M|.
\]
Therefore, as $\mathbb{E}|M|\to\infty$, we have w.h.p.
$
\bigl||M|-\mathbb{E}|M|\bigr|
=o\left(\mathbb{E}|M|\right).
$


\end{document}